\documentclass{article}

\PassOptionsToPackage{numbers, compress}{natbib}

\usepackage[final]{neurips_2024}

\usepackage[utf8]{inputenc} 
\usepackage[T1]{fontenc}    
\usepackage{hyperref}       
\usepackage{url}            
\usepackage{booktabs}       
\usepackage{amsfonts}       
\usepackage{nicefrac}       
\usepackage{microtype}      
\usepackage{xcolor}         

\usepackage{amssymb}
\usepackage{xspace}
\usepackage{enumitem}
\usepackage{graphicx}
\usepackage{amsmath}
\usepackage[capitalize,noabbrev,nameinlink]{cleveref}
\usepackage{bm}
\usepackage{colortbl}
\usepackage{scalerel}
\usepackage{multirow}
\usepackage{booktabs}
\usepackage{subcaption}
\usepackage{wrapfig}
\usepackage{amsthm}
\usepackage[thinc]{esdiff}

\definecolor{mydarkblue}{rgb}{0,0.08,0.45}
\definecolor{myblue}{HTML}{3b75c3}
\definecolor{myred}{HTML}{E33222}
\definecolor{mygreen}{HTML}{438773}
\definecolor{mymaroon}{RGB}{142,27,19}
\definecolor{maroon}{HTML}{800000}
\definecolor{mycite}{cmyk}{0.55,1,0,0.15}
\definecolor{codeblue}{rgb}{0.25,0.5,0.5}
\definecolor{codekw}{rgb}{0.85, 0.18, 0.50}
\definecolor{codegreen}{rgb}{0,0.6,0}
\definecolor{codegray}{rgb}{0.5,0.5,0.5}
\definecolor{codepurple}{rgb}{0.58,0,0.82}
\definecolor{backcolour}{rgb}{0.95,0.95,0.92}
\hypersetup{
    colorlinks=true,
    citecolor=codeblue,
    linkcolor=mymaroon,
    urlcolor=maroon
          }

\newcommand{\method}{\textsc{TAG-CF}\xspace}
\newcommand{\fullname}{\textbf{\underline{T}}est-time \textbf{\underline{Ag}}gregation for \textbf{\underline{C}}ollaborative \textbf{\underline{F}}iltering \xspace}

\newcommand{\methodp}{\textsc{TAG-CF$^+$}\xspace}
\newtheorem{theorem}{Theorem}
\newcommand\numberthis{\addtocounter{equation}{1}\tag{\theequation}}

\title{How Does Message Passing Improve \\ Collaborative Filtering?}

\author{%
  Clark Mingxuan Ju$^{1,2*}$, William Shiao$^{3}$, Zhichun Guo$^{2}$, Yanfang Ye$^{2}$,\\ 
  \textbf{Yozen Liu}$^{1}$, \textbf{Neil Shah}$^{1}$, \textbf{Tong Zhao}$^{1}$\thanks{Corresponding authors. Work done during first-author's internship at Snap Inc..} \\
  $^{1}$Snap Inc., $^{2}$University of Notre Dame, $^{3}$University of California, Riverside\\
  $^{1}$\texttt{\{mju,yliu2,nshah,tong\}@snap.com}, $^{2}$\texttt{\{mju2,zguo5,yye7\}@nd.edu}, $^{3}$\texttt{wshia002@ucr.edu} \\
}

\begin{document}

\maketitle

\begin{abstract}
Collaborative filtering (CF) has exhibited prominent results for recommender systems and been broadly utilized for real-world applications.
A branch of research enhances CF methods by message passing (MP) used in graph neural networks, due to its strong capabilities of extracting knowledge from graph-structured data, like user-item bipartite graphs that naturally exist in CF.
They assume that MP helps CF methods in a manner akin to its benefits for graph-based learning tasks in general (e.g., node classification). 
However, even though MP empirically improves CF, whether or not this assumption is correct still needs verification.
To address this gap, we formally investigate why MP helps CF from multiple perspectives and show that many assumptions made by previous works are not entirely accurate. 
With our curated ablation studies and theoretical analyses, we discover that \textit{\textbf{(i)} MP improves the CF performance primarily by additional representations passed from neighbors during the forward pass instead of additional gradient updates to neighbor representations during the model back-propagation and \textbf{(ii)} MP usually helps low-degree nodes more than high-degree nodes.}
Utilizing these novel findings, we present \fullname, namely \textbf{\method}, a test-time augmentation framework that only conducts MP once at inference time. 
The key novelty of \method is that it effectively utilizes graph knowledge while circumventing most of notorious computational overheads of MP. 
Besides, \method is extremely versatile can be used as a plug-and-play module to enhance representations trained by different CF supervision signals.
Evaluated on six datasets (i.e., five academic benchmarks and one real-world industrial dataset), \method consistently improves the recommendation performance of CF methods without graph by up to \textbf{39.2\%} on cold users and \textbf{31.7}\% on all users, with little to no extra computational overheads.
Furthermore, compared with trending graph-enhanced CF methods, \method delivers comparable or even better performance \textit{\underline{with less than \textbf{1\%} of their total training times}}.
Our code is publicly available at \url{https://github.com/snap-research/Test-time-Aggregation-for-CF}.
\end{abstract}

\section{Introduction}
Recommender systems are essential in improving users' experiences on web services, such as product recommendations~\citep{wang2021dcn,schafer1999recommender}, video recommendations~\citep{gomez2015netflix,van2013deep}, friend suggestions~\citep{sankar2021graph,ying2018graph}, etc.
In particular, recommender systems based on collaborative filtering (CF)
have shown superior performance~\citep{rendle2009bpr,koren2021advances,chen2023bias}. 
CF methods use preferences for items by users to predict additional topics or products a user might like~\citep{su2009survey}.
These methods typically learn a unique representation for each user/item and an item is recommended to a user according to their representation similarities~\citep{he2017neural,wang2015collaborative}.

One popular line of research explores Graph Neural Networks (GNNs) for CF, exhibiting improved results compared with CF frameworks without the utilization of graphs~\citep{wu2022graph,he2020lightgcn,wang2019neural,yu2022graph,cai2023lightgcl,wang2020disentangled,zhao2024learning}. 
The key mechanism behind GNNs is message passing (MP), where each node aggregates information from its neighbors in the graph, and information from neighbors that are multiple hops away can be acquired by stacked MP layers~\citep{kipf2016semi,velivckovic2017graph,hamilton2017inductive}.
During the model training, traditional CF methods directly fetch user/item representations of an observed interaction (e.g., purchase, friending, click, etc.) and enforce their pair-wise similarity~\citep{rendle2009bpr}. 
Graph-enhanced CF methods extend this scheme by conducting stacked MP layers over the user-item bipartite graph, and harnessing the resulting user and item representations to calculate a pair-wise affinity.

A recent study~\citep{he2020lightgcn} shows that removing several key components of the MP layer (e.g., learnable transformation parameters) greatly enhances GNNs' performance for CF.
Its proposed method (LightGCN) achieves promising performance by linearly aggregating neighbor representations and has been used as the de facto backbone model for later works due to its simple and effective design~\citep{cai2023lightgcl,yu2022graph,wu2021self}. 
However, this observation contradicts GNN architectures for classic graph learning tasks, where GNNs without these components severely under-perform~\citep{oloulade2021graph,wang2021bag}.
Additionally, existing research~\citep{he2020lightgcn,wang2019neural} assumes that the contribution of MP for CF is similar to that for graph learning tasks in general (e.g., node classification or link prediction)
- they posit that node representations are progressively refined by their neighbor information and the performance gain is positively proportional to the neighborhood density as measured in node degrees~\citep{tang2020investigating}. 
However, according to our empirical studies in \cref{sec:reason_degree}, MP in CF improves low-degree nodes more than high-degree nodes, which also contradicts GNNs' behaviors for classic tasks~\citep{tang2020investigating,hu2022tuneup}. 
In light of these inconsistencies, we ask: 
\begin{center}
\textbf{\textit{What role does message passing really play for collaborative filtering?}}
\end{center}
In this work, we investigate contributions brought by MP for CF from two perspectives. 
Firstly, we unroll the formulation of MP layer and show that its performance improvement could either come from additional representations passed from neighbors during the forward pass or accompanying gradient updates to neighbor representations during the back-propagation. 
With rigorously designed ablation studies, we empirically demonstrate that gains brought by the forward pass dominate those by the back-propagation. 
Furthermore, we analyze the performance distribution w.r.t. the user degree (i.e., the number of interactions per user) with or without message passing and discover that the message passing in CF improves low-degree users more compared to high-degree users.
For the first time, we connect this phenomenon to Laplacian matrix learning~\citep{zhu2021interpreting,dong2019learning,dong2016learning}, and theoretically show that popular supervision signals~\citep{rendle2009bpr,wang2022towards} for CF inadvertently conduct message passing in the back-propagation even without treating the input data as a graph.  
Hence, when message passing is applied, high-degree users demonstrate limited improvement, as the benefit of message passing for high degree nodes has already been captured by the supervision signal. 

With the above takeaways, we present \fullname, namely \textbf{\method}.
Specifically, unlike other graph CF methods, \method \emph{does not require any message passing during training}.
Instead, it is a test-time augmentation framework that only conducts a \emph{single message-passing step at inference time}, and effectively enhances representations inferred from different CF supervision signals. 
The test-time design is inspired by our first perspective that, within total performance gains brought by message passing, gains from the forward pass dominate those brought by the backward pass. 
Applying message passing only at test time avoids repetitive queries (i.e., once per node and epoch) for representations of surrounding nodes, which grow exponentially as the number of layers increases.
Moreover, following our second perspective that message passing helps low-degree nodes more in CF, we further offload the cost of \method by applying the one-time message passing only to low-degree nodes.
We summarize our contributions as:
\vspace{-0.06in}
\begin{itemize}[leftmargin=*]
    \item This is the first work that formally investigates why message passing helps collaborative filtering.
    We demonstrate that message passing in CF improves the performance primarily by additional representations passed from neighbors during the forward pass instead of accompanying gradient updates to neighbors during the back-propagation, and prove that message passing helps low-degree nodes more than high-degree nodes. 
    
    \item Given our findings, we propose \method, an efficient yet effective test-time aggregation framework to enhance representations inferred by different CF supervision signals such as BPR and DirectAU.
    Evaluated on six datasets, \method consistently improves the performance of CF methods without graph by up to \textbf{39.2\%} on cold users and \textbf{31.7}\% on all users, with little to no extra computational overheads.
    Furthermore, compared with trending graph-enhanced CF methods, \method delivers comparable or even better performance \textit{\underline{with less than \textbf{1\%} of their total training time}}.
    
    \item Beside promising cost-effectiveness, we show that test-time aggregation in \method improves the recommendation performance in similar ways as the training-time aggregation does, further demonstrating the legitimacy of our findings. 
\end{itemize}
\section{Preliminaries and Related Work}
\label{sec:pre}
\textbf{Collaborative Filtering}. 
Given a set of users, a set of items, and interactions between users and items, collaborative filtering (CF) methods aim at learning a unique representation for each user and item, such that user and item representations can reconstruct all observable interactions~\citep{rendle2009bpr,wang2022towards,koren2009matrix}.
CF methods based on matrix factorization directly utilize the inner product between a pair of user and item representations to infer the existence of their interaction~\citep{koren2009matrix,rendle2009bpr}. 
CF methods based on neural predictors use multi-layer feed-forward neural networks that take user and item representations as inputs and output prediction results~\citep{he2017neural,zhang2019deep}.
Let $\mathcal{U}$ and $\mathcal{I}$ denote the user set and item set respectively, with user $u_i \in \mathcal{U}$ associated with an embedding $\mathbf{u}_i \in \mathbb{R}^{d}$ and item $i_i \in \mathcal{U}$ associated with $\mathbf{i}_i \in \mathbb{R}^{d}$, the similarity $s_{ij}$ between user $u_i$ and item $i_j$ is formulated as $ s_{ij} = \hat{\mathbf{u}}_i^\intercal \cdot \hat{\mathbf{i}}_j$.

\noindent\textbf{Graph Neural Networks}. 
Graph neural networks (GNNs) are powerful learning frameworks to extract representative information from graphs~\citep{kipf2016semi,velivckovic2017graph,hamilton2017inductive,xu2018representation,fan2022heterogeneous,ju2022grape}, with numerous applications in large-scale ranking and forecasting tasks \citep{tang2020knowing, tang2022friend, derrow2021eta, shi2023embedding}.
They aim to map each input node into low-dimensional vectors, which can be utilized to conduct either graph-level~\citep{xu2018powerful} or node-level tasks~\citep{kipf2016semi}.  
Most GNNs explore layer-wise message passing~\citep{gilmer2017neural}, where each node iteratively extracts information from its first-order neighbors, and information from multi-hop neighbors can be captured by stacked layers.
Given a graph $\mathcal{G}=(\mathcal{V}, \mathcal{E})$ and node features $\mathbf{X} \in \mathbb{R}^{|\mathcal{V}|\times d}$, graph convolution in GCN~\citep{kipf2016semi} at $k$-th layer is formulated as:
\begin{equation}
\label{eq:gcn}
    \mathbf{h}_i^{(k+1)} = \sigma(\sum_{j \in \mathcal{N}(i)\cup{i}} \frac{1}{\sqrt{|N(i)|} \sqrt{|N(j)|}}\mathbf{h}_j^{(k)} \cdot \mathbf{W}^{(k)}),
\end{equation}
where $\mathbf{h}_i^{0} = \mathbf{x}_i$, $\mathcal{N}(i)$ refers to the set of direct neighbors of node $i$, and $\mathbf{W}^{(k)} \in \mathbb{R}^{d^k \times d^{(k+1)}}$ refers to parameters at the $k$-th layer transforming the node representation from $d^k$ to $d^{(k+1)}$ dimension.

Recent works~\citep{ma2021homophily,ma2021unified} have shown that GNNs make predictions based on the distribution of node neighborhoods.
Moreover, GNNs' performance improvement for high-degree nodes is typically better than for low-degree nodes~\citep{tang2020investigating,hu2022tuneup,ju2024graphpatcher,guo2024node, wang2024topological}.
They posit that node representations are progressively refined by their neighbor information and the performance gain is positively proportional to the neighborhood density as measured in node degrees. 
As we explore test-time augmentation in this work, it is worth noting that there also exist a group of relevant works that explore data augmentation techniques to enhance the GNN performance~\citep{ju2024graphpatcher,zhao2021data,zhao2022graph,ju2023let, jin2022empowering}. 

\noindent\textbf{Message Passing for Collaborative Filtering}.
Recent research tends to apply the message passing scheme in GNNs to CF~\citep{he2020lightgcn,wang2019neural,pal2020pinnersage,gao2022graph,shi2020heterogeneous,xia2021graph,kolodner2024robust,kung2024improving}. 
In CF, they mostly conduct message passing between user-item bipartite graphs and utilize the resultant representations to calculate user-item similarities.
For instance, NGCF~\citep{wang2019neural} directly migrates the message passing scheme in GNNs (similar to \cref{eq:gcn}) and applies it to bipartite graphs in CF. 
LightGCN~\citep{he2020lightgcn} simplifies NGCF~\citep{wang2019neural} by removing certain components (i.e., the self-loop, learning parameters for graph convolution, and activation functions) and further improves the recommendation performance compared with NGCF. 
The simplified parameter-less message passing in LightGCN can be expressed as: 
\begin{align}
    \mathbf{u}_i^{(k)} = \sum_{i_j \in N(u_i)} \frac{1}{\sqrt{|N(u_i)|} \sqrt{|N(i_j)|}} \mathbf{i}_j^{(k-1)},
    \mathbf{i}_i^{(k)} = \sum_{u_j \in N(i_i)} \frac{1}{\sqrt{|N(i_i)|} \sqrt{|N(u_j)|}} \mathbf{u}_j^{(k-1)},
    \label{eq:lgcn_1}
\end{align}
where $N(\cdot)$ refers to the set of items or users that the input interacts with, $\mathbf{u}_i^{(0)} = \mathbf{u}_i$, and $\mathbf{i}_i^{(0)} = \mathbf{i}_i$.
With $K$ layers, the final user/item representations and their similarities are constructed as:
\begin{equation}
    \hat{\mathbf{u}}_i = \frac{1}{K+1}\sum_{k=0}^K \mathbf{u}_i^{(k)}, \;\;
    \hat{\mathbf{i}}_i = \frac{1}{K+1}\sum_{k=0}^K \mathbf{i}_i^{(k)}, \;\;
    s_{ij} = \hat{\mathbf{u}}_i^\intercal \cdot \hat{\mathbf{i}}_j.
     \label{eq:lgcn_2}
\end{equation}
According to results reported in LightGCN and NGCF~\citep{he2020lightgcn,wang2019neural,chang2021sequential,gao2023survey} and empirical studies we provide in this work (i.e., \cref{tab:mftagcf} and \cref{tab:bprdau}), incorporating message passing to CF methods without graphs (i.e., matrix factorization methods~\citep{rendle2009bpr,he2017neural}) can improve the recommendation performance by up to 20\%. 
Utilizing LightGCN as the backbone model, later works try to further improve the performance by incorporating self-supervised learning signals~\citep{lin2022improving,yu2022graph,cai2023lightgcl,yu2021self,wei2021contrastive,ju2022multi}.
Graph-based CF methods assume that the contribution of message passing for CF is similar to that for graph learning tasks in general (e.g., node classification or link prediction).
However, whether or not this assumption is correct still needs verification, even though message passing empirically improves CF. 
There also exists a branch of research that aims at accelerating or simplifying message passing in CF by adding graph-based regularization terms during the training~\citep{shen2021powerful,mao2021ultragcn,peng2022svd,xia2023graph}. While promising, they still repetitively query representations of adjacent nodes during the training.

\textbf{Efficient Efforts in Matrix Factorization}. A branch of research specifically focuses on improving the efficiency of matrix factorization~\citep{shen2021powerful,peng2022svd,hong2024svd,park2024turbo,choi2023blurring}. For instance, GFCF~\citep{shen2021powerful} and Turbo-CF~\citep{park2024turbo} explore graph signal processing to linearly convolve the interaction matrix and use the resulted matrix directly for recommendation without training. Furthermore, SVD-GCN~\citep{peng2022svd} and SVD-AE~\citep{hong2024svd} utilize a low rank version of the interaction matrix to further accelerate the convolution efficiency and yet remain the promising performance. Besides, BSPM~\citep{choi2023blurring} studies using diffusion process to gradually reconstruct the interaction matrix and achieves promising performance with fast processing. In parallel with these existing efforts, we propose to enhance any existing matrix factorization method through test-time augmentation that harnesses graph-based heuristics.
\section{How Does Message Passing Improve Collaborative Filtering?} \label{sec:howmp} 
In this section, we demonstrate why message passing (MP) helps collaborative filtering from two major perspectives: 
Firstly, we focus on inductive biases brought by the MP explored in LightGCN, the de facto backbone model for graph-based CF methods. 
Secondly, we consider the performance improvement on different node subgroups w.r.t. the node degree with and without MP. 

\vspace{-0.05in}
\subsection{Neighbor Information vs. Accompanying Gradients from Message Passing} \label{sec:reason_mp}
Following the definition in \cref{eq:lgcn_1}, given a one-layer LightGCN\footnote{For the simplicity of the notation, we showcase our observation with only one layer. However, since LightGCN is fully linear, the phenomenon we show also applies to variants with arbitrary layers.}, we unroll the calculation of the similarity $s_{ij}$ between any user $u_i$ and item $i_j$  as the following:
\begin{align*}
    & s_{ij}  =  \Big(\mathbf{u}_i + \sum_{i_n \in N(u_i)} \frac{1}{\sqrt{|N(u_i)|} \sqrt{|N(i_n)|}} \mathbf{i}_n\Big)^\intercal \cdot \Big(\mathbf{i}_j + \sum_{u_n \in N(i_j)} \frac{1}{\sqrt{|N(i_j)|} \sqrt{|N(u_n)|}} \mathbf{u}_n\Big) \\
    & \;\;\; =\; \mathbf{u}_i^\intercal \cdot \mathbf{i}_j + 
    \sum_{u_n \in N(i_j)} \frac{1}{\sqrt{|N(i_j)|} \sqrt{|N(u_n)|}} \mathbf{u}_i^\intercal\cdot\mathbf{u}_n 
     \sum_{i_n \in N(u_i)} \frac{1}{\sqrt{|N(u_i)|} \sqrt{|N(i_n)|}} \mathbf{i}_n^\intercal\cdot \mathbf{i}_j \\
    & \;\;\; + \sum_{i_n \in N(u_i)} \sum_{u_n \in N(i_j)} \frac{1}{\sqrt{|N(u_i)|} \sqrt{|N(i_n)|} \sqrt{|N(i_j)|} \sqrt{|N(u_n)|}} \mathbf{i}_n^\intercal\cdot \mathbf{u}_n.\numberthis \label{eq:lgcn_unroll}
\end{align*}
With derived similarities between user-item pairs, their corresponding representations can be updated by objectives (e.g., BPR~\citep{rendle2009bpr} and DirectAU~\citep{wang2022towards}) that enforce the pair-wise similarity between representations of user-item pairs in the training data. 

CF methods without the utilization of graphs directly calculate the similarity between a user and an item with their own representations (i.e.,  $s_{ij}=\mathbf{u}_i^\intercal \cdot \mathbf{i}_j$), which aligns with the first term in \cref{eq:lgcn_unroll}.
Compared to the formulation in \cref{eq:lgcn_unroll}, we can see that three additional similarity terms are introduced as inductive biases: similarities between users who purchase the same item (i.e., $\mathbf{u}_i^\intercal\cdot\mathbf{u}_n$), between items that share the same buyer (i.e., $\mathbf{i}_n^\intercal\cdot \mathbf{i}_j$), and between neighbors of an observed interaction (i.e., $\mathbf{i}_n^\intercal\cdot \mathbf{u}_n$). 
With these three additional terms from MP, we reason that the performance improvement brought by MP to CF methods without graph could come from {\bf (i)} additional neighbor representations during the forward pass (i.e., numerical values of three extra terms in \cref{eq:lgcn_unroll}), or {\bf (ii)} accompanying gradient updates to neighbors during the back-propagation.

To investigate the origin of the performance improvement brought by MP, we designed two variants of LightGCN. 
The first one (LightGCN$_\text{w/o neigh. info}$) shares the same forward and backward procedures as LightGCN during the training
but does not conduct MP during the test time. 
In this variant, additional gradients brought by MP are maintained as part of the resulting model, but information from neighbors are ablated. 
In the second variant (LightGCN$_\text{w/o grad.}$), the model shares the same forward pass but drops gradients from these three additional terms during the backward propagation. 
Besides these two variants, we also experiment on LightGCN without MP, denoted as LightGCN$_\text{w/o both}$, a matrix factorization model with the same supervision signal (i.e., BPR loss).
Implementation details w.r.t. this experiment are in \cref{app:ep}. 
\begin{wraptable}{r}{.5\textwidth}
    \vspace{-0.1in}
    \captionof{table}{Performance of LightGCN variants.} 
    \vspace{-0.1in}
    \label{tab:contvsgrad}
    \resizebox{.5\columnwidth}{!}{
    \begin{tabular}{c|ccc}
    \toprule
    Method & \texttt{Yelp-2018} & \texttt{Gowalla} & \texttt{Amazon-book} \\ 
    \midrule
    \multicolumn{4}{c}{NDCG@20} \\
    \cmidrule(r){1-4} 
    LightGCN         & 6.36  & 9.88 & 8.13 \\
    \;\; w/o grad.   & 6.16 \;(3.1\%$\downarrow$) & 9.87 \;(0.1\%$\downarrow$) & 7.80 \;(4.1\%$\downarrow$)\\
    \;\; w/o neigh. info   & 4.71 (25.9\%$\downarrow$) & 6.95 (29.7\%$\downarrow$) & 6.95 (14.5\%$\downarrow$)\\
    \;\; w/o both    & 6.09 \;(4.2\%$\downarrow$) & 9.83 \;(0.5\%$\downarrow$) & 7.75 \;(4.7\%$\downarrow$)\\
    \cmidrule(r){1-4} 
    \multicolumn{4}{c}{Recall@20} \\
    \cmidrule(r){1-4} 
    LightGCN         & 11.21  & 18.53 & 12.97 \\
    \;\; w/o grad.   & 10.87 \;(3.0\%$\downarrow$) & 18.51 \;(0.1\%$\downarrow$) & 12.81 \;(1.2\%$\downarrow$)\\
    \;\; w/o neigh. info   & 8.44 (24.7\%$\downarrow$) & 13.06 (29.5\%$\downarrow$) & 11.25 (13.3\%$\downarrow$)\\
    \;\; w/o both    & 10.71 \;(4.5\%$\downarrow$) & 18.42 \;(0.6\%$\downarrow$) & 12.57 \;(3.1\%$\downarrow$)\\
    \bottomrule
    \end{tabular}
    }
    \vspace{-0.1in}
\end{wraptable}
From \cref{tab:contvsgrad}, we observe that the performance of all variants 
is downgraded compared with LightGCN,
with the most significant degradation on LightGCN$_\text{w/o neigh. info}$. 
This phenomenon indicates that {\bf (i)} both additional representations passed from neighbors during the forward pass and accompanying gradient updates to neighbors during the back-propagation help the recommendation performance, and {\bf (ii)} within total performance gains brought by MP, gains from the forward pass dominate those brought by the back-propagation. 
Comparing LightGCN with LightGCN$_\text{w/o grad.}$, we notice that the incorporation of gradient updates brought by MP is relatively incremental (i.e., $\sim$2\%).
However, to facilitate these additional gradient updates for slightly better performance, LightGCN is required to conduct MP at each batch, which brings tremendous additional overheads. 

\subsection{Message Passing in CF Helps Low-degree Users More Compared with High-degrees} \label{sec:reason_degree}

Both empirical and theoretical evidence have demonstrated that GNNs usually perform satisfactorily on high-degree nodes with rich neighbor information but not as well on low-degree nodes~\citep{tang2020investigating,hu2022tuneup}.
While designing graph-based model architectures for CF, most existing methods directly borrow this line of observations~\citep{wang2019neural,he2020lightgcn} and assume that the contribution of message passing for CF is similar to that for graph learning tasks in general.
However, whether or not these observations still transfer to message passing in CF remains questionable, as there exist architectural and philosophical gaps between message passing for CF and its counterparts for GNNs, as discussed in \cref{sec:pre}.
To validate these hypotheses, we conduct experiments over representative methods (i.e., LightGCN and matrix factorization (MF) trained with BPR) and show their performance w.r.t. the node degree in \cref{fig:degree}.

We observe that, overall both MF and LightGCN perform better on high-degree users than low-degree users. 
According to the upper two figures in \cref{fig:degree}, MF behaves similarly to LightGCN, even without treating the input data as graphs, where the overall performance for high-degree user is stronger than that for low-degree users.
However, the performance improvement of LightGCN from MF on low-degree users is larger than that for high-degree users (i.e., lower two figures in \cref{fig:degree}).
According to literature in general graph learning tasks~\citep{hu2022tuneup,liu2021tail,tang2020investigating}, the performance improvement should be positively proportional to the node degree - the gain for high-degree users should be higher than that for low-degree users.
This discrepancy indicates that it might not be appropriate to accredit contributions of message passing in CF directly through ideologies designed for classic graph learning tasks (e.g., node classification and link prediction).
\begin{figure}[t]
    \vspace{-0.4in}
    \resizebox{1.0\columnwidth}{!}{
    \includegraphics{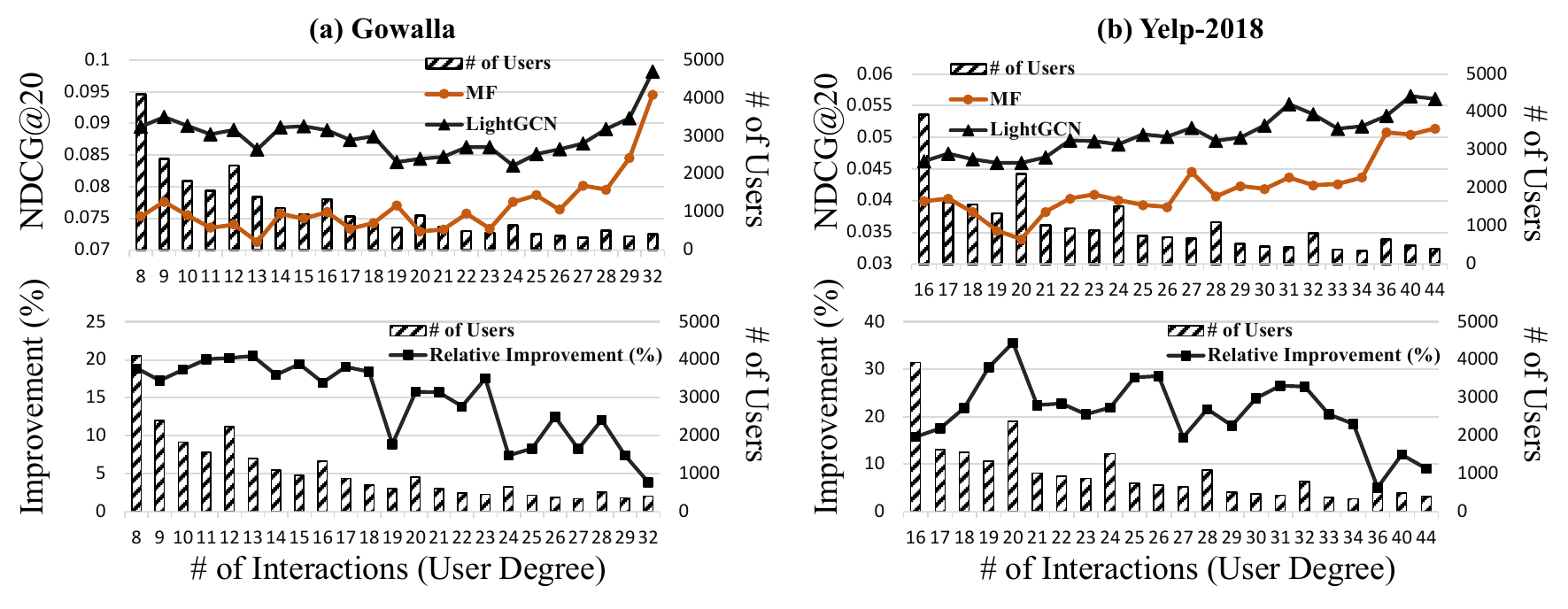}
    }
    \vspace{-0.3in}
    \caption{Performances of LightGCN and
Matrix Factorization w.r.t. the user degree across datasets.
The performance improvement brought by message passing decreases as the user degree goes up.
}
    \label{fig:degree}
    \vspace{-0.15in}
\end{figure}
To bridge this gap, we connect supervision signals (i.e., BPR and DirectAU) commonly adopted by CF methods to Laplacian matrix learning. 
The formulation of BPR~\citep{rendle2009bpr} and DirectAU~\citep{wang2022towards} without the incorporation of graphs can be written as:
\begin{align*}
     \mathcal{L}_{\text{BPR}} = - \sum_{(i,j)\in\mathcal{D}}\sum_{(i,k)\notin\mathcal{D}}\log\sigma(s_{ij}-s_{ik})= 
     -\sum_{(i,j)\in\mathcal{D}}\sum_{(i,k)\notin\mathcal{D}}\log\sigma(\mathbf{u}_i^\intercal\cdot\mathbf{i}_j - \mathbf{u}_i^\intercal\cdot\mathbf{i}_k), \\
    \mathcal{L}_{\text{DirectAU}}  =  \sum_{(i,j)\in\mathcal{D}} ||\mathbf{u}_i - \mathbf{i}_j||^2 + \sum_{u, u' \in \mathcal{U}}\log e^{-2||\mathbf{u}-\mathbf{u}'||^2}  + 
    \sum_{i, i' \in \mathcal{I}}\log e^{-2||\mathbf{i}-\mathbf{i}'||^2},\numberthis
\end{align*}
where $\mathcal{D}$
refers to the set of observed interactions at the training phase and $ \mathbf{i'}$ and $ \mathbf{u'}$ refers to any random user/item.
According to works on Laplacian matrix learning~\citep{zhu2021interpreting,dong2019learning, ma2021unified}, learning node representations over graphs can be decoupled into Laplacian quadratic form, a weighted summation of two sub-goals:
\begin{equation}
    \min_{\mathbf{Z}}\{||\mathbf{Z}-\mathbf{X}||^2+\text{tr}(\mathbf{Z}^\intercal\mathbf{L}\mathbf{Z})\},
    \label{eq:laplacian}
\end{equation}
where $\mathbf{Z}$ refers to the node representation matrix after the message passing, $\mathbf{X}$ refers to the input feature matrix, and $\mathbf{L}$ refers to the Laplacian matrix. 
The first term regularizes the latent representation such that it does not diverge too much from the input feature; whereas the second term promotes the similarity between latent representations of adjacent nodes, which can be re-written as:
$\text{tr}(\mathbf{Z}^\intercal\cdot\mathbf{L}\cdot\mathbf{Z}) = \sum_{(i,j)\in \mathcal{D}} ||\mathbf{u}_i - \mathbf{i}_j||^2$ in CF bipartite graphs. \cite{zhu2021interpreting} show that $K$ layers of linear message passing exactly optimizes the second term in \cref{eq:laplacian}.
Given this theoretical foundation, we derive the following theorem w.r.t. relations between BPR, DirectAU, and message passing in CF:
\begin{theorem}
\label{thm:bprdau}
Assuming that $||\mathbf{u}_i||^2 = ||\mathbf{i}_j||^2 =1$ for any $u_i \in \mathcal{U}$ and  $I_j \in \mathcal{I}$, objectives of BPR and DirectAU are strictly upper-bounded by the objective of message passing (i.e., $\mathcal{L}_{\text{BPR}} \leq \sum_{(i,j)\in \mathcal{D}} ||\mathbf{u}_i - \mathbf{i}_j||^2$ and $\mathcal{L}_{\text{DirectAU}} \leq \sum_{(i,j)\in \mathcal{D}} ||\mathbf{u}_i - \mathbf{i}_j||^2$). 
\end{theorem}
Proof of \cref{thm:bprdau} can be found in \cref{app:proof}. 
According to \cref{thm:bprdau}, both BPR and DirectAU optimize the objective of message passing (i.e., $\sum_{(i,j)\in \mathcal{D}} ||\mathbf{u}_i - \mathbf{i}_j||^2$) with some additional regularization (i.e., dissimilarity between non-existing user/item pairs for BPR, and representation uniformity for DirectAU).
Hence, directly optimizing these two objectives partially fulfills the effects brought by message passing during the back-propagation. 

Combining this theory with the aforementioned empirical observations, we show that these two supervision signals could inadvertently conduct message passing in the backward step, even without explicitly treating interaction data as graphs. 
Since this inadvertent message passing happens during the back-propagation, its performance is positively correlated to the amount of training signals a user/item can get. In the case of CF, the amount of training signals for a user is directly proportional to the node degree. High-degree active users naturally benefit more from the inadvertent message passing from objective functions, because they acquire more training signals. Hence, when explicit message passing is applied to CF methods, the performance gain for high-degree users is less significant than that for low-degree users. Because the contribution of the message passing over high-degree nodes has been mostly fulfilled by the inadvertent message passing during the training. 

To quantitatively prove this theory, we incrementally upsample low-degree training users and observe the performance improvement that TAG-CF could introduce at each upsampling rate.
If our line of theory is correct, then we should expect less performance improvement on low-degree users for a larger upsampling rate. 
The results are shown in \cref{sec:upsample} with supporting evidence.

\section{Test-time Aggregation for Collaborative Filtering}
In \cref{sec:howmp}, we demonstrate why message passing helps CF from two perspectives. Firstly, w.r.t. the formulation of LightGCN, we observe that the performance gain brought by neighbor information dominates that brought by additional gradients. 
Secondly, w.r.t. the improvement on user subgroups, we learn that message passing helps low-degree users more, compared with high-degree users. 

In light of these two takeaways, we present \textbf{\underline{T}}est-time \textbf{\underline{Ag}}gregation for \textbf{\underline{C}}ollaborative \textbf{\underline{F}}iltering, namely \textbf{\method}, a test-time augmentation framework that only conducts message passing once at inference time and is effective at enhancing  matrix factorization methods
trained by different CF supervision signals.
Given a set of well-trained user/item representations, \method simply aggregates neighboring item (user) representations for a given user (item) at test time.  
Despite its simplicity, we show that our proposal can be used as a plug-and-play module and is effective at enhancing representations trained by different CF supervision signals.

The test-time aggregation is inspired by our first perspective that, within total performance gains brought by message passing, gains from additional neighbor representations during the forward pass dominate those brought by accompanying gradient updates to neighbors during the back-propagation.
Applying message passing only at test time 
avoids repetitive training-time queries (i.e., once per node and epoch) of surrounding neighbors, which grow exponentially as the number of layers increases by the neighbor explosion phenomenon \citep{guo2023linkless, zhang2021graph, zeng2021decoupling}.
Specifically, given a set of well-trained user and item representations $\mathbf{U} \in \mathbb{R}^{|\mathcal{U}| \times d}$ and $\mathbf{I}\in \mathbb{R}^{|\mathcal{I}| \times d}$, \method augments representations for user $u_i$ and item $i_i$ as:
\begin{equation}\label{eq:tagcf}
    \begin{aligned}
         &\mathbf{u}_i^* = \mathbf{u}_i + \sum_{i_j \in N(u_i)} |N(u_i)|^m |N(i_j)|^n \cdot \mathbf{i}_j, 
     \mathbf{i}_i^* = \mathbf{i}_i + \sum_{u_j \in N(i_i)} |N(i_i)|^m |N(u_j)|^n \cdot \mathbf{u}_j,
    \end{aligned}
\end{equation}
where $m$ and $n$ are two hyper-parameters that control the normalization of message passing. 
With $m=n=-\frac{1}{2}$, \cref{eq:tagcf} becomes the exact formulation of one-layer LightGCN (i.e., \cref{eq:lgcn_1}). 
Empirically, we observe that the setup with $m=n=-\frac{1}{2}$ for \method does not always work for all datasets.
This setup is directly migrated from message passing for homogeneous graphs~\citep{kipf2016semi}, which might not be applicable for bipartite graphs where all neighbors are heterogeneous~\citep{dasoulas2020learning}.
Unlike LightGCN which can fill this gap by adaptively tuning all representations during the training, \method cannot update any parameter since it is applied at test time, and hence requires tune-able normalization hyper-parameters.

Moreover, following our second perspective that message passing helps low-degree nodes more in CF, we further derive \methodp, which reduces the cost of \method by applying the one-time message passing only to low-degree nodes. 
Focusing on only low-degree nodes has two benefits: {\bf (i)} it reduces the number of nodes that \methodp needs to attend to, and {\bf (ii)} message passing for low-degree nodes is naturally cheaper than for high-degree nodes given the surrounding neighborhoods are sparser (mitigating neighbor explosion). 
The degree threshold that determines which nodes to apply \methodp is selected by the validation performance, with details in \cref{app:ep}. 

\method can effectively enhance MF methods by conducting message passing only once at test time. 
\method effectively utilizes graphs while circumventing most of notorious computational overheads of message passing. 
It is extremely flexible, simple to implement, and enjoys the performance benefits of graph-based CF method while paying the lowest overall scalability.
\vspace{-0.1in}
\section{Experiments}
\vspace{-0.1in}
We conduct extensive experiments to demonstrate the effectiveness and efficiency of \method. 
We aim to answer the following research questions: 
\textbf{RQ (1)}: how effective is \method at improving MF methods without using graphs, 
\textbf{RQ (2)}: how much overheads does \method introduce, 
\textbf{RQ (3)}: can \method effectively enhance MF methods trained by different objectives,
\textbf{RQ (4)}: how effective is \methodp w.r.t. different degree cutoffs,
and \textbf{RQ (5)}: do behaviors of \method align with our findings in \cref{sec:howmp}?

\vspace{-0.1in}
\subsection{Experimental Settings}
\vspace{-0.1in}
\textbf{Datasets}. 
We conduct comprehensive experiments on five commonly used academic benchmark datasets,
including \texttt{Amazon-book}, \texttt{Anime}, \texttt{Gowalla}, \texttt{Yelp2018}, and \texttt{MovieLens-1M}. Additionally,
we also evaluate on a large-scale real-world industrial user-item recommendation dataset \texttt{Internal}.
Descriptions of these datasets are provided in \cref{app:dataset}.

\noindent\textbf{Baselines}.
We compare \method with two branches of methods:
(1) CF methods that do not explicitly utilize graphs, including vanilla matrix factorization (MF) methods trained by BPR and DirectAU~\citep{rendle2009bpr,wang2022towards}, Efficient Neural Matrix Factorization~\citep{chen2020efficient} (denoted as ENMF), and UltraGCN~\cite{mao2021ultragcn}.
(2) Graph-based CF methods, 
including LightGCN~\citep{he2020lightgcn} and NGCF~\citep{wang2019neural}. 
Besides, we also compare with
recent graph-based CF methods that extend LightGCN by adding additional self-supervised signals for better performance, including LightGCL~\citep{cai2023lightgcl}, SimGCL~\citep{yu2022graph}, and SGL~\citep{wu2021self}. 
For the coherence of reading, we include comprehensive discussions about evaluation protocols across all methods, tuning for hyper-parameters, and other implementation details in \cref{app:ep}.

\vspace{-0.1in}
\subsection{Performance Improvement to Matrix Factorization Methods} 
\vspace{-0.1in}
For \textbf{RQ (1)}, \cref{tab:mftagcf} shows the performances of MF methods (MF and ENMF) as well as that of the performances of them with \method applied on their learned representations.
We observe that \method unanimously improves the recommendation performance for both of them. 
Specifically, across all datasets, \method on average improves the low-degree NDCG@20 by 4.6\% and 9.1\% and overall NDCG by 3.2\% and 7.1\% for ENMF and MF, respectively.
We also observe a similar performance improvement for Recall@20, where \method on average improves the low-degree Recall@20 by 4.5\% and 8.4\% and overall Recall@20 by 2.1\% and 7.2\% for ENMF and MF, respectively.
Furthermore, we notice that \method can improve the performance of UltraGCN, a method that utilizes the graph knowledge as additional supervision signals.
This phenomenon demonstrates the superior effectiveness of \method, indicating that our proposed test-time aggregation can further enhance graph-enhanced MF methods. 

By comparing the performance gains brought by \method on low-degree users with that on all users, we notice that gains for low-degree users are usually higher.
Hence, message passing in CF helps low-degree users more than for high-degree users, which echos with our observations in \cref{sec:reason_mp}.
To answer \textbf{RQ (5)}, the behavior of \method aligns with our second perspective in \cref{sec:reason_degree} that the supervision signal inadvertently conducts message passing.
Consequently, the room for improvement on high-degree users could be limited, as part of the contributions from message passing has already been claimed by the supervision signal.
\begin{table*}
\begin{center}
\caption{Recommendation performance (i.e., NDCG@20 and Recall@20) of all models across users with different numbers of interactions. 
The lower percentile indicates the set of nodes whose degrees are ranked in the lower 30\% population.
\textbf{Bold} and \underline{underline} indicate the best and second best model respectively.
LightGCN and MF are trained with DirectAU~\citep{wang2022towards}.
} 
\resizebox{1.0\columnwidth}{!}{
\begin{tabular}{l|cc|cc>{\columncolor[gray]{0.9}}c|cc>{\columncolor[gray]{0.9}}c|cc>{\columncolor[gray]{0.9}}c}
\toprule
Method & NGCF & LightGCN & ENMF & +TAG-CF & Impr. ($\uparrow$) & MF & +TAG-CF & Impr. ($\uparrow$\%) & UltraGCN & +TAG-CF & Impr. ($\uparrow$\%)\\
\midrule
\multicolumn{12}{c}{\textsc{NDCG@20 -- Low-degree Users (Lower Percentile)}} \\
\cmidrule(r){1-12} 
\texttt{Amazon-Book}& 5.32$_{\pm{0.08}}$ & \underline{8.09}$_{\pm{0.10}}$ & 5.33$_{\pm{0.02}}$ & 5.67$_{\pm{0.03}}$ & 6.4\%& 8.02$_{\pm{0.07}}$ & \textbf{8.26}$_{\pm{0.06}}$ & 3.0\% & 5.61$_{\pm{0.19}}$ & 6.04$_{\pm{0.21}}$ & 7.7\%\\
\texttt{Anime} & 20.13$_{\pm{0.18}}$ & 27.78$_{\pm{0.21}}$ & 22.23$_{\pm{0.19}}$ & 22.58$_{\pm{0.15}}$ & 1.6\% & 23.95$_{\pm{0.07}}$ & 27.15$_{\pm{0.04}}$ & 13.4\% & \underline{28.14}$_{\pm{0.19}}$ & \textbf{30.10}$_{\pm{0.21}}$ & 7.0\%\\
\texttt{Gowalla} & 8.46$_{\pm{0.06}}$ & \underline{10.08}$_{\pm{0.13}}$ & 3.87$_{\pm{0.15}}$ & 4.08$_{\pm{0.11}}$ & 5.4\% & 10.00$_{\pm{0.08}}$ & \textbf{10.19}$_{\pm{0.04}}$ & 1.9\% & 8.21$_{\pm{0.09}}$ & 8.63$_{\pm{0.11}}$ & 5.1\%\\
\texttt{Yelp-2018}  & 4.87$_{\pm{0.06}}$ & \underline{6.10}$_{\pm{0.09}}$ & 3.11$_{\pm{0.07}}$ & 3.26$_{\pm{0.04}}$ & 4.8\% & 6.08$_{\pm{0.08}}$ & \textbf{6.18}$_{\pm{0.05}}$ & 1.7\% & 4.89$_{\pm{0.10}}$ & 5.44$_{\pm{0.12}}$ & 11.2\% \\
\texttt{MovieLens-1M} & 22.13$_{\pm{0.26}}$ & 25.95$_{\pm{0.28}}$ & 18.34$_{\pm{0.19}}$ & 22.53$_{\pm{0.21}}$ & 22.8\% & 20.98$_{\pm{0.12}}$ & \textbf{29.20}$_{\pm{0.19}}$ & 39.2\% & 23.89$_{\pm{0.19}}$ & \underline{28.37}$_{\pm{0.21}}$ & 18.8\% \\
\texttt{Internal} & 5.91$_{\pm{0.07}}$ & \underline{8.12}$_{\pm{0.03}}$ & OOM & - & - & 6.79$_{\pm{0.04}}$ & \textbf{8.52}$_{\pm{0.06}}$ & 25.5\% & OOM & - & -\\
\cmidrule(r){1-12} 
\multicolumn{12}{c}{\textsc{NDCG@20 -- Overall}} \\
\cmidrule(r){1-12} 
\texttt{Amazon-Book} & 6.97$_{\pm{0.11}}$ & \underline{8.06}$_{\pm{0.11}}$ & 6.13$_{\pm{0.13}}$ & 6.54$_{\pm{0.09}}$ & 6.7\% & 8.01$_{\pm{0.03}}$ & \textbf{8.13}$_{\pm{0.03}}$ & 1.5\% & 5.77$_{\pm{0.25}}$  & 6.11$_{\pm{0.27}}$ & 5.9\%\\
\texttt{Anime} & 22.54$_{\pm{0.25}}$ & 27.97$_{\pm{0.21}}$ & 30.17$_{\pm{0.09}}$ & \underline{30.86}$_{\pm{0.12}}$ & 2.3\% & 24.01$_{\pm{0.06}}$ & 27.25$_{\pm{0.03}}$ & 9.8\% & 30.30$_{\pm{0.11}}$ & \textbf{30.89}$_{\pm{0.11}}$ & 1.9\% \\
\texttt{Gowalla} & 8.65$_{\pm{0.10}}$ & \textbf{9.96}$_{\pm{0.11}}$ & 5.23$_{\pm{0.04}}$ & 5.29$_{\pm{0.05}}$ & 1.1\% & 9.77$_{\pm{0.08}}$ & \underline{9.88}$_{\pm{0.04}}$ & 1.1\% & 8.53$_{\pm{0.14}}$ & 9.02$_{\pm{0.15}}$ & 5.7\% \\
\texttt{Yelp-2018} & 5.54$_{\pm{0.06}}$ & \underline{6.33}$_{\pm{0.06}}$ & 3.79$_{\pm{0.09}}$ & 3.89$_{\pm{0.05}}$ & 2.6\% &6.25$_{\pm{0.06}}$ & \textbf{6.36}$_{\pm{0.03}}$ & 1.8\% & 5.01$_{\pm{0.11}}$ & 5.53$_{\pm{0.11}}$ & 10.4\%\\
\texttt{MovieLens-1M} & 23.17$_{\pm{0.18}}$ & 26.64$_{\pm{0.23}}$ & 20.57$_{\pm{0.18}}$ & 22.98$_{\pm{0.20}}$ & 11.7\% & 22.51$_{\pm{0.14}}$ & \underline{29.65}$_{\pm{0.17}}$ & 31.7\% & 26.50$_{\pm{0.15}}$& \textbf{29.68}$_{\pm{0.21}}$ & 12.0\% \\
\texttt{Internal} & 6.94$_{\pm{0.06}}$ & \underline{8.10}$_{\pm{0.06}}$ & OOM & - & - & 7.04$_{\pm{0.02}}$ & \textbf{8.54}$_{\pm{0.02}}$ & 21.3\% & OOM & - & -\\
\midrule
\midrule
\multicolumn{12}{c}{\textsc{Recall@20 -- Low-degree Users (Lower Percentile)}} \\
\cmidrule(r){1-12} 
\texttt{Amazon-Book}& 10.71$_{\pm{0.14}}$ & \underline{13.18}$_{\pm{0.17}}$ & 10.42$_{\pm{0.16}}$& 11.08$_{\pm{0.11}}$& 6.3\% & 13.07$_{\pm{0.09}}$ & \textbf{13.37}$_{\pm{0.10}}$ & 2.3\% & 7.92$_{\pm{0.15}}$ & 8.31$_{\pm{0.10}}$ & 4.9\%\\
\texttt{Anime} & 25.74$_{\pm{0.35}}$ & 32.74$_{\pm{0.21}}$ & \underline{37.14}$_{\pm{0.59}}$ & \textbf{38.41}$_{\pm{0.53}}$ & 3.4\% & 29.08$_{\pm{0.09}}$ & 31.94$_{\pm{0.05}}$ & 9.8\% & 33.96$_{\pm{0.28}}$ & 36.49$_{\pm{0.28}}$ & 7.4\% \\
\texttt{Gowalla} & 17.53$_{\pm{0.32}}$ & \underline{19.14}$_{\pm{0.20}}$ & 8.73$_{\pm{0.08}}$ & 9.01$_{\pm{0.06}}$ & 3.2\% & 18.92$_{\pm{0.19}}$ & \textbf{19.17}$_{\pm{0.13}}$ & 1.3\% & 15.57$_{\pm{0.18}}$ & 16.01$_{\pm{0.15}}$ & 2.8\%\\
\texttt{Yelp-2018}  & 10.15$_{\pm{0.13}}$ & \underline{10.75}$_{\pm{0.14}}$ & 7.17$_{\pm{0.06}}$ & 7.54$_{\pm{0.12}}$ & 5.2\% & 10.63$_{\pm{0.13}}$ & \textbf{10.98}$_{\pm{0.14}}$ & 3.3\% & 7.71$_{\pm{0.15}}$ & 8.59$_{\pm{0.18}}$ & 11.4\%\\
\texttt{MovieLens-1M} & 22.71$_{\pm{0.16}}$ & 25.80$_{\pm{0.22}}$ & 19.58$_{\pm{0.14}}$ & 24.11$_{\pm{0.16}}$ & 23.1\% & 23.64$_{\pm{0.18}}$ & \underline{28.10}$_{\pm{0.20}}$ & 18.9\% & 26.13$_{\pm{0.21}}$ & \textbf{28.97}$_{\pm{0.23}}$ & 10.9\%\\
\texttt{Internal} & 10.54$_{\pm{0.09}}$ & \underline{13.81}$_{\pm{0.02}}$ & OOM & - & - & 11.13$_{\pm{0.05}}$ & \textbf{13.97}$_{\pm{0.06}}$ & 25.5\% & OOM & - & -\\
\cmidrule(r){1-12} 
\multicolumn{12}{c}{\textsc{Recall@20 -- Overall}} \\
\cmidrule(r){1-12} 
\texttt{Amazon-Book} & 10.30$_{\pm{0.21}}$ & \underline{12.76}$_{\pm{0.18}}$ & 10.89$_{\pm{0.18}}$ & 11.35$_{\pm{0.09}}$ & 4.2\% & 12.67$_{\pm{0.06}}$ & \textbf{12.97}$_{\pm{0.06}}$ & 2.4\% & 8.01$_{\pm{0.25}}$ & 8.53$_{\pm{0.27}}$ & 6.5\%\\
\texttt{Anime} & 28.12$_{\pm{0.22}}$ & 32.82$_{\pm{0.21}}$ & 34.10$_{\pm{0.25}}$ & 34.48$_{\pm{0.23}}$ & 1.1\% & 29.15$_{\pm{0.09}}$ & 31.95$_{\pm{0.05}}$ & 6.9\% & \underline{35.87}$_{\pm{0.39}}$ & \textbf{37.01}$_{\pm{0.39}}$ & 3.2\%\\
\texttt{Gowalla} & 17.93$_{\pm{0.06}}$ & \textbf{18.65}$_{\pm{0.14}}$ & 9.68$_{\pm{0.06}}$ & 9.74$_{\pm{0.09}}$ & 0.6\% & 18.30$_{\pm{0.17}}$ & \underline{18.53}$_{\pm{0.11}}$ & 1.3\% & 15.93$_{\pm{0.21}}$ & 16.36$_{\pm{0.22}}$ & 2.7\% \\
\texttt{Yelp-2018} & 10.02$_{\pm{0.06}}$ & \underline{10.98}$_{\pm{0.10}}$ & 6.89$_{\pm{0.09}}$ & 7.05$_{\pm{0.03}}$ & 2.3\% & 10.81$_{\pm{0.10}}$ & \textbf{11.21}$_{\pm{0.09}}$ & 3.7\% & 8.41$_{\pm{0.19}}$ & 9.89$_{\pm{0.20}}$ & 17.6\% \\
\texttt{MovieLens-1M} & 23.93$_{\pm{0.14}}$ & 26.30$_{\pm{0.20}}$ & 21.31$_{\pm{0.19}}$ & 23.88$_{\pm{0.25}}$ & 12.1\% & 26.30$_{\pm{0.14}}$ & \underline{28.40}$_{\pm{0.15}}$ & 8.0\% & 27.14$_{\pm{0.19}}$ & \textbf{29.78}$_{\pm{0.23}}$ & 9.7\% \\
\texttt{Internal} & 6.91$_{\pm{0.04}}$ & \underline{13.89}$_{\pm{0.06}}$ & OOM & - & - & 11.83$_{\pm{0.02}}$ & \textbf{14.41}$_{\pm{0.08}}$ & 21.8\% & OOM & - & -\\
\bottomrule
\end{tabular}
}
\label{tab:mftagcf}
\end{center}
\end{table*}

\begin{table*}
\begin{center}
\caption{Running time ($1 \times 10^3$ seconds) for MF methods and \method. Time \% is the percentage of running time \method takes w.r.t. the time for corresponding MF methods.
Speed$\uparrow$ refers to the ratio of running times between training-time aggregation (i.e., LightGCN) and \method.
All training steps are timed and terminated by an early stopping strategy (see \cref{app:ep}).
}
\resizebox{1.0\columnwidth}{!}{
    \begin{tabular}{l|c|cc>{\columncolor[gray]{0.9}}c|cc>{\columncolor[gray]{0.9}}c|ccc>{\columncolor[gray]{0.9}}c>{\columncolor[gray]{0.9}}c}
    \toprule
    Method  & Sparsity  & ENMF & +TAG-CF & Time \% & UltraGCN & +TAG-CF & Time \% & LightGCN & MF & +TAG-CF & Time \% & Speed$\uparrow$\\
    \midrule
    \texttt{Anime} &99.13\%  & 12.31 & +0.04 & 0.3\% & 93.31 & +0.04 & 0.1\% & 138.85 & 34.12 & +0.04 & 0.3\% & 4.06\texttimes{}\\
    \texttt{Yelp-2018} & 99.87\% & 2.15 & +0.02 & 0.9\% & 5.02 & +0.02 & 0.4\% & 5.81 & 3.17 & +0.02 & 0.6\% & 1.83\texttimes{}\\
    \texttt{Gowalla} & 99.91\% & 4.56 & +0.02 & 0.4\% & 12.55 &  +0.02 & 0.2\% & 13.27 & 7.74 & +0.02 & 0.3\% & 1.72\texttimes{}\\
    \texttt{Amazon-Book} & 99.94\% & 11.54 & +0.03 & 0.3\% & 39.25 & +0.03 & 0.1\% & 46.62 & 29.21 & +0.03 & 0.1\% & 1.59\texttimes{}\\
    \texttt{Internal} & 99.99\% & OOM & - & - & OOM & - & - & 47.32 & 32.62 & + 0.09 & 0.3\% & 1.44\texttimes{}\\
    \bottomrule
    \end{tabular}
    }
\label{tab:time}
\end{center}
\end{table*}
\begin{table}[t]
\begin{minipage}[t]{.51\textwidth}
\vspace{-0.43in}
\captionof{table}{The running time and performance of graph-based CF methods that extend LightGCN.} 
\resizebox{1.\columnwidth}{!}{
    \begin{tabular}{l|cccc}
    \toprule
    Method  & SGL & SimGCL & LightGCL & TAG-CF\\
    \midrule
    \multicolumn{5}{c}{\textsc{NDCG@20 -- Overall}} \\
    \midrule
    \texttt{Anime} & 27.02$_{\pm{0.05}}$ & 30.48$_{\pm{0.12}}$ & 28.34$_{\pm{0.16}}$ & 27.25$_{\pm{0.03}}$ \\
    \texttt{Yelp} & 5.67$_{\pm{0.04}}$ & 5.99$_{\pm{0.09}}$ & 4.93$_{\pm{0.06}}$ & 6.36$_{\pm{0.03}}$ \\
    \texttt{Gowalla} & 9.67$_{\pm{0.17}}$  & 10.32$_{\pm{0.06}}$ & 8.99$_{\pm{0.13}}$ & 9.88$_{\pm{0.04}}$ \\
    \texttt{Book} & 6.69$_{\pm{0.02}}$ & 7.02$_{\pm{0.05}}$ & 5.83$_{\pm{0.08}}$ & 8.13$_{\pm{0.03}}$ \\
    \rowcolor[gray]{0.9} Avg. Rank & 3.2 & \underline{1.7} & 3.5 & \textbf{1.2} \\
    \midrule
    \multicolumn{5}{c}{\textsc{Recall@20 -- Overall}} \\
    \midrule
    \texttt{Anime} & 31.29$_{\pm{0.09}}$ & 34.93$_{\pm{0.14}}$ & 33.64$_{\pm{0.22}}$ & 31.95$_{\pm{0.05}}$ \\
    \texttt{Yelp} & 10.01$_{\pm{0.08}}$ & 10.56$_{\pm{0.13}}$ & 8.83$_{\pm{0.04}}$ & 11.21$_{\pm{0.09}}$ \\
    \texttt{Gowalla} & 18.18$_{\pm{0.24}}$ & 19.22$_{\pm{0.09}}$ & 16.99$_{\pm{0.10}}$ & 18.53$_{\pm{0.09}}$ \\
    \texttt{Book} & 11.15$_{\pm{0.04}}$ & 11.51$_{\pm{0.09}}$ & 10.06$_{\pm{0.05}}$ & 12.97$_{\pm{0.06}}$ \\
    \rowcolor[gray]{0.9} Avg. Rank & 3.2 & \textbf{1.5} & 3.5 & \underline{1.7} \\
    \midrule
    \multicolumn{5}{c}{\textsc{Running Time ($1 \times 10^3$ second)}} \\
    \midrule
    \texttt{Anime} & 69.48 & 87.77 & 97.31 & 34.15 \\
    \texttt{Yelp}& 3.94 & 9.72 & 4.30 & 3.19 \\
    \texttt{Gowalla} & 9.32 & 29.11 & 11.10 & 7.76 \\
    \texttt{Book}& 63.21 & 71.39 & 38.87 & 29.24 \\
    \rowcolor[gray]{0.9} Avg. Rank & \underline{2.2} & 3.8 & 3.0 & \textbf{1.0} \\
    \midrule
    \rowcolor[gray]{0.9} Total Rank & 3.6 & \underline{2.8} & 3.9 & \textbf{1.9} \\
    \bottomrule
    \end{tabular}
    }
\label{tab:graph}
\end{minipage} \hfill
\hspace{0.2cm}
\begin{minipage}[t]{.48\textwidth}\
\vspace{-0.32in}
\captionof{table}{Performance of \method when applied to models trained with BPR loss.} 
\vspace{-0.1in}
\resizebox{1.\columnwidth}{!}{
    \begin{tabular}{l|c|cc>{\columncolor[gray]{0.9}}c}
    \toprule
    Method  & LightGCN & MF & +TAG-CF & Impr. ($\uparrow$\%)\\
    \midrule
    \multicolumn{5}{c}{\textsc{NDCG@20 -- Low-degree Users (Lower Percentile)}} \\
    \midrule
    \texttt{Anime} & 30.02$_{\pm{0.07}}$ & 29.36$_{\pm{0.23}}$ & 30.56$_{\pm{0.27}}$ & 4.1\% \\
    \texttt{Yelp}  & 4.34$_{\pm{0.07}}$ & 3.63$_{\pm{0.15}}$ & 3.81$_{\pm{0.18}}$ & 5.0\% \\
    \texttt{Gowalla} & 8.22$_{\pm{0.03}}$ & 7.56$_{\pm{0.14}}$ & 7.88$_{\pm{0.15}}$ & 4.2\% \\
    \texttt{Book} & 5.19$_{\pm{0.14}}$ & 4.19$_{\pm{0.14}}$ & 4.68$_{\pm{0.14}}$ & 11.7\%\\
    \midrule
    \multicolumn{5}{c}{\textsc{NDCG@20 -- Overall}} \\
    \midrule
    \texttt{Anime} & 30.14$_{\pm{0.07}}$ & 29.51$_{\pm{0.21}}$ & 30.23$_{\pm{0.26}}$ & 2.4\% \\
    \texttt{Yelp} & 4.87$_{\pm{0.06}}$ & 3.96$_{\pm{0.14}}$ & 4.26$_{\pm{0.17}}$ & 7.6\% \\
    \texttt{Gowalla} & 8.32$_{\pm{0.03}}$ & 7.51$_{\pm{0.12}}$ & 7.99$_{\pm{0.14}}$ & 6.4\% \\
    \texttt{Book} & 5.07$_{\pm{0.15}}$ & 4.15$_{\pm{0.13}}$ & 4.32$_{\pm{0.13}}$ & 4.1\% \\
    \midrule
    \multicolumn{5}{c}{\textsc{Recall@20 -- Low-degree Users (Lower Percentile)}} \\
    \midrule
    \texttt{Anime} & 34.23$_{\pm{0.08}}$ & 34.81$_{\pm{0.32}}$ & 35.42$_{\pm{0.35}}$ & 1.8\%\\
    \texttt{Yelp} & 8.19$_{\pm{0.20}}$ & 6.93$_{\pm{0.26}}$ & 7.25$_{\pm{0.19}}$ & 4.6\% \\
    \texttt{Gowalla} & 16.17$_{\pm{0.12}}$ & 14.86$_{\pm{0.23}}$ & 15.33$_{\pm{0.24}}$ & 3.2\% \\
    \texttt{Book} & 8.81$_{\pm{0.26}}$ & 7.45$_{\pm{0.22}}$ & 8.05$_{\pm{0.15}}$ & 8.1\% \\ 
    \midrule
    \multicolumn{5}{c}{\textsc{Recall@20 -- Overall}} \\
    \midrule
    \texttt{Anime} & 34.21$_{\pm{0.08}}$ & 34.84$_{\pm{0.30}}$ & 35.23$_{\pm{0.34}}$ & 1.1\%\\
    \texttt{Yelp} & 8.33$_{\pm{0.30}}$ & 7.27$_{\pm{0.27}}$ & 7.62$_{\pm{0.22}}$ & 4.8\%\\
    \texttt{Gowalla} & 15.69$_{\pm{0.07}}$ & 14.47$_{\pm{0.23}}$ & 14.92$_{\pm{0.25}}$ & 3.1\% \\
    \texttt{Book} & 8.65$_{\pm{0.24}}$ & 7.35$_{\pm{0.22}}$ & 7.64$_{\pm{0.20}}$ & 3.9\%\\ 
    \bottomrule
    \end{tabular}
    }
\label{tab:bprdau}
\end{minipage}%
\vspace{-0.15in}
\end{table}
\vspace{-0.1in}
\subsection{Performance Comparison Among Graph-based Methods}
\vspace{-0.1in}
Comparing \method with LightGCN in \cref{tab:mftagcf}, we can notice that \method mostly performs on par with and sometimes even outperforms LightGCN, without incorporating message passing during the training and only conducting test-time aggregation.
This phenomenon indicates that conducting neighbor aggregation at the testing time can recover most of the contributions of training-time message passing. 
To answer \textbf{RQ (5)}, \method aligns with our first perspective in \cref{sec:reason_mp} that the performance gain from beneficial neighbor information dominates their accompanying gradients.

We further compare \method with state-of-the-art graph-based CF methods, with their performance and efficiency shown in \cref{tab:graph}.
Among these performant baselines, \method exhibits competitive performance, with an average rank of 1.2 on NDCG and 1.7 on Recall.
Though not always the model that delivers the best performance, \method can deliver comparably promising results and introduces little computational overheads (i.e., ranked 1.0 for running time).
Considering efficiency as one factor, \method achieves the best performance across all baselines with an average rank of 1.9.

While performing on par with graph-based CF methods that aggregate neighbor contents at the training time, \method enjoys the performance benefits of message passing while paying the lowest overall scalability.
To answer \textbf{RQ (2)}, according to \cref{tab:time}, across all datasets, \method only introduces an average additional computational overhead of $0.05 \times 10^3$ seconds, which is less than 0.5\% of the total training time for matrix factorization methods. 
Comparing the running time of LightGCN with that of \method, we can observe that the latter can significantly improve the computational time, and the speedup is proportional to the sparsity of the dataset. 

\subsection{Effectiveness for Different Training Signals}
To answer \textbf{RQ (3)}, besides DirectAU, we also conduct experiments on BPR loss, as shown in \cref{tab:bprdau}.
When applied to BPR, \method still consistently improves the performance by large margins (i.e., 6.3\% and 5.1\% average improvement on low-degree and overall NDCG respectively, and 4.4\% and 3.2\% on low-degree and overall Recall respectively).
We notice that \method sometimes does not perform as competitively as LightGCN when both are trained with BPR.
We check norms of learned representations from MF with BPR and discover that they have high variance since BPR does not explicitly enforce any regularization. 
This might not favor \method as a test-time augmentation method due to its simple design, which cannot adapt representations with high variance. 

\subsection{Performance w.r.t. User Degree}

To answer \textbf{RQ (4)}, we apply \methodp to four public datasets and the performance and the efficiency improvement are demonstrated in \cref{fig:plus}.
Overall, the running time improvement brought by \methodp exponentially increases as the degree decreases, since low-degree users have sparse neighborhoods and there is hence less information for \methodp to aggregation. 
When the degree cutoff is low (i.e., less than 100), the effectiveness of \methodp proportional increases as the degree cutoff increases. 

When setting the cutoff to a user degree of around 100,
on \texttt{Amazon-Book}, \texttt{Gowalla}, and \texttt{Yelp-2018}, \methodp can further improve \method by 125\%, 17\%, and 11\%, respectively, with efficiency improvement of 7\%, 4\%, and 8\%.
In these cases, \methodp not only significantly improves the performance but also effectively reduces computational overheads. 
\begin{wrapfigure}{r}{0.7\textwidth} 
  \begin{center}
    \vspace{-0.2in}
    \includegraphics[width=0.7\textwidth]{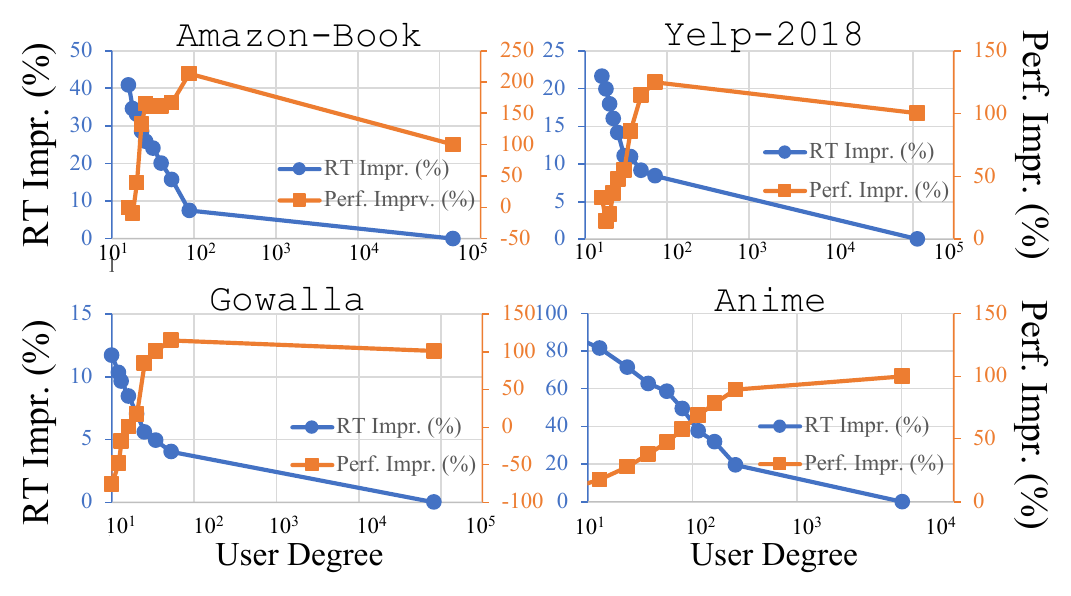}
    \caption{The performance and efficiency improvement of \methodp w.r.t. different cutoffs.
    \methodp further improves \method with less computational overheads. 100\% is the original performance/efficiency of vanilla \method.
    }
    \vspace{-0.2in}
    \label{fig:plus}
  \end{center}
\end{wrapfigure}
However, on these three datasets, after the cutoff bypasses a degree of 100, the performance improvement eventually decreases to the performance of \method (i.e., 100\%), indicating that test-time aggregation jeopardizes the performance on high-degree nodes. 
On \texttt{Anime}, though no downgrade on high-degree users, the performance improvement of \methodp to \method is incremental. 
These phenomenons not only demonstrate the effectiveness and efficiency of \methodp, but also verify our findings in \cref{sec:reason_degree} that message passing in CF helps low-degree users more than high-degree users.

\section{Conclusion}

In this study, we investigate how message passing improves collaborative filtering. 
Through a series of ablations, we demonstrate that the performance gain from neighbor contents dominates that from accompanying gradients brought by message passing in CF. 
Moreover, for the first time, we show that message passing in CF improves low-degree users more than high-degree users.
We theoretically demonstrate that CF supervision signals inadvertently conduct message passing in the backward step, even without treating the data as a graph. 
In light of these novel takeaways, we propose TAG-CF, a test-time aggregation framework effective at enhancing representations trained by different CF supervision signals.
Evaluated on five datasets, \method performs at par with SoTA methods with only a fraction of computational overhead (i.e., less than 1.0\% of the total training time).

\section{Limitation and Broader Impact}
One limitation of our proposal could be the utilization of graphs in large-scale machine learning pipeline. 
TAG-CF conducts a single-time aggregation of neighbors, which could be equivalently achieved by existing technologies such as SQL, BigQuery, etc. 
Furthermore, we observe no ethical concern entailed by our proposal, but we note that both ethical or unethical applications based on collaborative filtering may benefit from the effectiveness of our work. 
Care should be taken to ensure socially positive and beneficial results of machine learning algorithms.

\section{Acknowledgments}
This work was mostly conducted during the internship of Clark, William, and Zhichun at Snap Inc. 
We would like to thank Xin Chen and his colleagues from Snap Inc. for their help on pre-processing the internal dataset. 
This work was partially supported by the NSF under grants IIS-2321504, IIS-2334193, IIS-2203262, IIS-2217239, CNS-2426514, CNS-2203261, and CMMI-2146076. Any opinions, findings, and conclusions or recommendations expressed in this material are those of the authors and do not necessarily reflect the views of the sponsors. 

\noindent We sincerely appreciate constructive feedback from all reviewers during the paper review phase. 

\bibliographystyle{abbrv}
\bibliography{reference}
\appendix
\vspace{-0.1in}
\section{Proof of Theorem 1} \label{app:proof}
Here we re-state \cref{thm:bprdau} before diving into its proof:

\textbf{Theorem 1.}
 Given that a K-layer GCN minimizes $ \sum_{(i,j)\in \mathcal{D}} ||\mathbf{u}_i - \mathbf{i}_j||^2$, 
During the training of MF methods, assuming that $||\mathbf{u}_i||^2 = ||\mathbf{i}_j||^2 =1$ for any $u_i \in \mathcal{U}$ and  $I_j \in \mathcal{I}$, objectives of BPR and DirectAU are strictly upper-bounded by the objective of message passing (i.e., $\mathcal{L}_{\text{BPR}} \leq \sum_{(i,j)\in \mathcal{D}}$ $||\mathbf{u}_i - \mathbf{i}_j||^2$ and $\mathcal{L}_{\text{DirectAU}} \leq \sum_{(i,j)\in \mathcal{D}} ||\mathbf{u}_i - \mathbf{i}_j||^2$). 

One preliminary theoretical foundation for \cref{thm:bprdau} to hold is that a K-layer graph convolution network (GCN) exactly optimizes the second term in \cref{eq:laplacian}, which has been proved by \cite{zhu2021interpreting}.
For ease of reading, we re-phrase it again as the following:
\begin{theorem}
\label{thm:kgcn}
The message passing for GCN optimizes the following graph regularization term: $\mathcal{O}=\min_\mathbf{Z}\{tr(\mathbf{Z}^\intercal\mathbf{L}\mathbf{Z}))\}$.
\end{theorem}
\begin{proof}
    Set derivative of $tr(\mathbf{Z}^\intercal\mathbf{L}\mathbf{Z})$ with respect to $\mathbf{Z}$ to zero:
    \begin{equation}
    \diffp{tr(\mathbf{Z}^\intercal\mathbf{L}\mathbf{Z})}{{\mathbf{Z}}} = 0
    \rightarrow \mathbf{L}\mathbf{Z}=0 \rightarrow \mathbf{Z} = A\mathbf{Z}.
    \label{eq:pf1}
    \end{equation}
    With K$\rightarrow\infty$:
    \begin{equation}
        \mathbf{Z}^{(K)} = \mathbf{A}\mathbf{Z}^{(K-1)}
    \end{equation}
    which indicates:
    \begin{equation}
        \mathbf{Z}^{(K)} =\mathbf{A}\mathbf{Z}^{(K-1)} = \mathbf{A}^2\mathbf{Z}^{(K-2)} = \dots=\mathbf{A}^K\mathbf{Z}^{(0)} = \mathbf{A}^K\mathbf{X}\mathbf{W}.
    \end{equation}
\end{proof}
According to this theoretical foundation, it is straightforward that \cref{thm:kgcn} is also applicable for the message passing of LightGCN in the setting of CF if we let $\mathbf{A} = \{0,1\}^{(|\mathcal{U}|+|\mathcal{I}|) \times (|\mathcal{U}|+|\mathcal{I}|)}$, $\mathbf{X} = \mathbf{I}_{|\mathcal{U}|+|\mathcal{I}|}$, and $\mathbf{W} = (\mathbf{U}||\mathbf{I}$), where $||$ refers to the concatenation operation.
With this preliminary, the proof to \cref{thm:bprdau} starts as:
\begin{proof}
    DirectAU optimizes: 
    \begin{align}
        \mathcal{L}_{\text{DirectAU}} & = \sum_{(i,j)\in\mathcal{D}} ||\mathbf{u}_i - \mathbf{i}_j||^2 \\
        & + \sum_{u, u' \in \mathcal{U}}\log e^{-2||\mathbf{u}-\mathbf{u}'||^2} + 
        \sum_{i, i' \in \mathcal{I}}\log e^{-2||\mathbf{i}-\mathbf{i}'||^2}.
    \end{align}
    Since $\sum_{u, u' \in \mathcal{U}}\log e^{-2||\mathbf{u}-\mathbf{u}'||^2} <= 0$ and  $\sum_{i, i' \in \mathcal{I}}\log e^{-2||\mathbf{i}-\mathbf{i}'||^2}<=0$, we directly have $\mathcal{L}_{\text{DirectAU}} \leq \sum_{(i,j)\in \mathcal{D}} ||\mathbf{u}_i - \mathbf{i}_j||^2$.

    BPR optimizes: 
    \begin{align}
        \mathcal{L}_{\text{BPR}} &= - \sum_{(i,j)\in\mathcal{D}}\sum_{(i,k)\notin\mathcal{D}}\log\sigma(s_{ij}-s_{ik})= \\
        & -\sum_{(i,j)\in\mathcal{D}}\sum_{(i,k)\notin\mathcal{D}}\log\sigma(\mathbf{u}_i^\intercal\cdot\mathbf{i}_j - \mathbf{u}_i^\intercal\cdot\mathbf{i}_k) \\
        & = \sum_{(i,j)\in\mathcal{D}}\sum_{(i,k)\notin\mathcal{D}} -\log\Big(\frac{e^{\mathbf{u}_i^\intercal\cdot\mathbf{i}_j}}{e^{\mathbf{u}_i^\intercal\cdot\mathbf{i}_j} + e^{\mathbf{u}_i^\intercal\cdot\mathbf{i}_k}}\Big) \\
        & =
        \sum_{(i,j)\in\mathcal{D}}\sum_{(i,k)\notin\mathcal{D}} -\mathbf{u}_i^\intercal\cdot\mathbf{i}_j + \log\Big( e^{\mathbf{u}_i^\intercal\cdot\mathbf{i}_j} + e^{\mathbf{u}_i^\intercal\cdot\mathbf{i}_k}\Big)
        \label{eq:intermediate}
    \end{align}
    Since $||\mathbf{u}_i||^2 = ||\mathbf{i}_j||^2 =1$ for any $u_i \in \mathcal{U}$ and  $I_j \in \mathcal{I}$, $||\mathbf{u}_i - \mathbf{i}_j|| = \sqrt{1 - 2\mathbf{u}_i^\intercal \cdot \mathbf{i}_j + 1} \rightarrow -\mathbf{u}_i^\intercal\cdot\mathbf{i}_j = \frac{1}{2}||\mathbf{u}_i - \mathbf{i}_j||^2 - 1$. So \cref{eq:intermediate} can be written as:
    \begin{equation}
         \mathcal{L}_{\text{BPR}} = \frac{1}{2}||\mathbf{u}_i - \mathbf{i}_j||^2 - 1 + \log\Big( e^{\mathbf{u}_i^\intercal\cdot\mathbf{i}_j} + e^{\mathbf{u}_i^\intercal\cdot\mathbf{i}_k}\Big). 
    \end{equation}
    The maximum possible value of $e^{\mathbf{u}_i^\intercal\cdot\mathbf{i}_j} + e^{\mathbf{u}_i^\intercal\cdot\mathbf{i}_k}$ is $2e$, which is less than 10. Hence $\log\Big( e^{\mathbf{u}_i^\intercal\cdot\mathbf{i}_j} + e^{\mathbf{u}_i^\intercal\cdot\mathbf{i}_k}\Big) < 1$, which leads to the second part of \cref{thm:bprdau}: $\mathcal{L}_{\text{BPR}} \leq \sum_{(i,j)\in \mathcal{D}} ||\mathbf{u}_i - \mathbf{i}_j||^2$.
\end{proof}

\section{Dataset Description and Statistics}
\label{app:dataset}
We conduct experiments on five commonly used benchmark datasets,
that have been broadly utilized by the recommender system community, 
including \texttt{Amazon-book}~\citep{mcauley2016addressing}, \texttt{Anime}~\citep{anime}, \texttt{Gowalla}~\citep{cho2011friendship},  \texttt{Yelp2018} \citep{yelp}, and \texttt{MovieLens-1M}~\citep{harper2015movielens}.
Additionally, we also evaluate our method on a large-scale industrial user-content recommendation dataset - \texttt{Internal}, with statistics shown in \cref{tab:dataset}.
\begin{table}[h]
    \centering
    \caption{Statistics of datasets explored in this work. Due to privacy constrains, we only report approximated values for \texttt{Internal} dataset.}
    \begin{tabular}{l|cccc}
    \toprule 
    Dataset & \# Users & \# Items & \# Interactions & Sparsity \\
    \cmidrule(r){1-5} 
     \texttt{Amazon-book}            & 52,643  & 40,981  & 2,984,108 & 99.94\% \\
     \texttt{Anime}        & 73,515  & 12,295  & 7,813,727 & 99.13\%  \\
     \texttt{Gowalla}          & 29,858 & 40,981  & 1,027,370 & 99.91\%  \\
     \texttt{Yelp-2018}         & 31,668 & 38,048 & 1,561,406 & 99.87\%  \\
     \texttt{MovieLens-1M}         & 6,040 & 3,629 & 836,478 & 96.18\%  \\
     \texttt{Internal}    &  $\sim$0.5M & $\sim$0.2M & $\sim$7M & 99.99\% \\
    \bottomrule
    \end{tabular}
    \label{tab:dataset}
\end{table}
\section{Additional Experimental Settings}
\label{app:ep}
\subsection{Evaluation Protocol} We evaluate all models using metrics adopted in previous works, including NDCG@20 and Recall@20~\citep{he2020lightgcn}.
For the dataset split, we conduct the group-by-user splits and randomly select 80\%, 10\%, and 10\% of observed interactions as training, validation, and testing sets respectively. We adopt an early stopping strategy, where the training will be terminated if the validation NDCG@20 stops increasing for 3 continuous epochs. 
We use models with the best validation performance to report the performance. 
Besides, the evaluation metrics are computed by the all-ranking protocol, where all items are listed as candidates~\citep{rendle2009bpr}.
We explore this strategy since we want to evaluate the representation quality of all users. 
All experiments are conducted 10 times with different seeds, and we report both means and standard deviations across independent runs.

\subsection{Hyper-parameter Tuning}
We only conduct 25 searches per model for all methods to ensure the comparison fairness, so that our experiments are not biased to methods with sophisticated hyper-parameter search spaces.
Furthermore, we set the embedding dimensions for all models to 64 (i.e., $d=64$) to ensure a fair comparison, since a larger dimension usually leads to better performance in CF methods.
For \method, we only tune $m$ and $n$ in \cref{eq:tagcf} during test time from the list of [-2, -1.5, -1, -0.5, 0].
Besides, we train all models using Adam optimizer. TAG-CF's sensitivty to $m$ and $n$ is visually plotted in \cref{fig:mandn}.
We can observe that m and n are important for the success of TAG-CF. Fortunately, across datasets, the optimal selection of m and n is pretty similar (e.g., m=n=-0.5 or m=n=0). The other solution to automatically tune m and n could be initialzing m and n to -0.5 (i.e., the value that generally works well across datasets) and conducting gradient descent on them using the training loss. But in this work we observe that manually tuning them on a small set of candidates can already deliver promising results.

\subsection{Implementation Detail}
We conduct most of the baseline experiments with RecBole~\citep{zhao2022recbole}.
Besides, we use Google Cloud Platform with 12 CPU cores, 64GB RAM, and a single V100 GPU with 16GB VRAM to run all experiments. 
\begin{table}
  \begin{center}
    \captionof{table}{Improvement of \methodp to \method. Degree cutoffs are selected according to \cref{fig:cutoff}.} 
    \label{tab:plus}
    \begin{tabular}{l|cccc}
    \toprule
    Metric & \texttt{Yelp-2018} & \texttt{Gowalla} & \texttt{Amazon-book} & \texttt{Anime} \\ 
    \midrule
    \multicolumn{5}{c}{\textsc{BPR}} \\
    \cmidrule(r){1-5} 
    NDCG@20 & 27.1\% & 10.3\% & 122.4\% & 0\% \\
    Recall@20 & 31.4\% & 14.2\% & 119.2\% & 0\% \\
    Running Time & 8\% & 4\% & 9\% & 0\% \\
    \cmidrule(r){1-5} 
    \multicolumn{5}{c}{\textsc{DirectAU}} \\
    \cmidrule(r){1-5} 
    NDCG@20 & 34.1\% & 22.5\% & 98.3\% & 0\%\\
    Recall@20 & 29.2\% & 30.1\% & 104.1\% & 0\% \\
    Running Time & 8\% & 4\% & 9\% & 0\% \\
    \bottomrule
    \end{tabular}
  \end{center}
\end{table}
\section{Degree Cutoff Selection for TAG-CF+}
We first sort all users according to their degree and split the sorted list into 10 user buckets\footnote{The number of buckets can be set to arbitrary numbers for finer adjustments. In this study, we pick 10 as a proof of concept.}, where each bucket contains non-overlapped users with similar degrees.
Starting from the bucket with the lowest user degree, \methodp keeps applying test-time-aggregation demonstrated in \cref{eq:tagcf} to all buckets until the validation performance starts to decrease or the performance improvement is less than $2\%$ compared with \method.
The degree cutoffs circled in \cref{fig:cutoff} are the ones selected by this strategy and most of them correspond to the most performant configuration, shown in \cref{tab:plus}.
\label{app:cutoff}
\begin{figure}
    \includegraphics[width=1\textwidth]{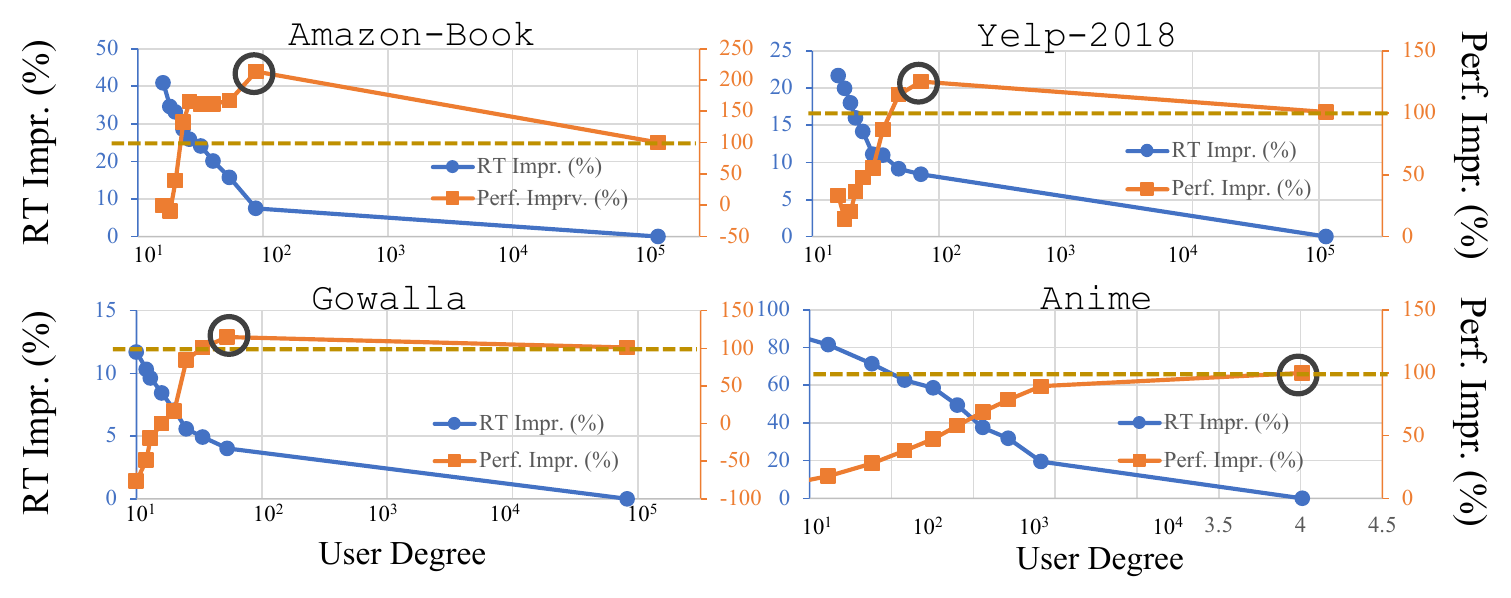}
    \caption{
Improvement of \methodp w.r.t. different cutoffs.
Yellow dashed lines indicate \method, and black circles refer to the optimal degree cutoff that \methodp selects.
}
    \label{fig:cutoff}
    \vspace{-0.2in}
\end{figure}

\begin{wraptable}{r}{.6\textwidth}
  \begin{center}
  \vspace{-0.2in}
    \captionof{table}{Improvement (NDCG@20) brought by TAG-CF at different degree cutoffs and upsampling rates on ML-1M.} 
    \vspace{-0.1in}
    \label{tab:upsample}
    \resizebox{0.6\columnwidth}{!}{
    \begin{tabular}{c|cc>{\columncolor[gray]{0.9}}c|cc>{\columncolor[gray]{0.9}}c}
    \toprule
    Up-sampling & \multicolumn{3}{c}{\multirow{2}{*}{Up-sample Rate: 100\%}} & \multicolumn{3}{c}{\multirow{2}{*}{Up-sample Rate: 300\%}} \\
    Degree  \\
    \midrule
    & MF & + \method & Impr. ($\uparrow$\%) & MF & + \method & Impr. ($\uparrow$\%) \\
    \midrule
    40 & 20.62 &28.87 & 38.8\% & 19.30 &25.01 & 30.3\% \\
    80 & 20.10 & 27.43 & 35.9\% & 18.40 &23.30 & 26.8\%\\
    160& 19.39 & 26.63 & 36.6\% & 17.93 &23.37 & 29.8\% \\
    \bottomrule
    \end{tabular}
    }
  \end{center}
  \vspace{-0.2in}
\end{wraptable}

\section{Analysis of TAG-CF through Up-sampling} \label{sec:upsample}
In \cref{sec:howmp}, we connect CF objective functions to message passing and show that they inadvertently conduct message passing during the back-propagation. 
Since this inadvertent message passing happens during the back-propagation, its performance is positively correlated to the amount of training signals a user/item can get. 
In the case of CF, the amount of training signals for a user is directly proportional to the node degree of this user. 
High-degree active users naturally benefit more from the inadvertent message passing from objective functions like BPR and DirectAU, because they acquire more training signals from the objective function. 
Hence, when explicit message passing is applied to CF methods, the performance gain for high-degree users is less significant than that for low-degree users. 
Because the contribution of the message passing over high-degree nodes has been mostly fulfilled by the inadvertent message passing during the training. 

To quantitatively prove this line of theory, we incrementally up-sample low-degree training examples and observe the performance improvement that TAG-CF could introduce at each upsampling rate. 
If our line of theory is correct, then we should expect less performance improvement on low-degree users for a larger upsampling rate. 
The results are shown in \cref{tab:upsample}.
From this table, though upsampling low-degree users hurts the overall performance, we can observe that the performance improvement brought by TAG-CF for low-degree users decreases, as the upsampling rate increases. 

According to this experiment, we can conclude that the more supervision signals a user receives (no matter for a low-degree or high-degree user), the less performance improvement message passing can bring. 
This experiment quantitatively shows why the performance improvement of high-degree users could be limited more than low-degree users. 
Because high-degree users naturally receive more training signals during the training whereas low-degree users receive fewer training signals. 

\section{Experiments on Ranking Metrics@10 } \label{sec:10metrics}
This section shows the performance of all models as well as \method's improvement to them when evaluated with ranking metrics with 10 candidates. The results are shown in \cref{tab:rebuttal_table1}, with similar trends as we have observed in \cref{tab:mftagcf}.

\begin{table*}
\begin{center}
\caption{Recommendation performance (i.e., NDCG@10 and Recall@10) of all models across users with different numbers of interactions. 
Setting explored in this table is the same as what \cref{tab:mftagcf} has. 
} 
\resizebox{1.0\columnwidth}{!}{
\begin{tabular}{l|cc|cc>{\columncolor[gray]{0.9}}c|cc>{\columncolor[gray]{0.9}}c|cc>{\columncolor[gray]{0.9}}c}
\toprule
Method & NGCF & LightGCN & ENMF & +TAG-CF & Impr. ($\uparrow$) & MF & +TAG-CF & Impr. ($\uparrow$\%) & UltraGCN & +TAG-CF & Impr. ($\uparrow$\%)\\
\midrule
\multicolumn{12}{c}{\textsc{NDCG@10 -- Low-degree Users (Lower Percentile)}} \\
\cmidrule(r){1-12} 
\texttt{Amazon-Book} & 3.33$_{\pm{0.09}}$ & 5.05$_{\pm{0.11}}$ & 3.32$_{\pm{0.02}}$ & 3.54$_{\pm{0.03}}$ & 6.8\%& 5.01$_{\pm{0.08}}$ & 5.19$_{\pm{0.06}}$ & 3.7\%& 3.49$_{\pm{0.17}}$ & 3.80$_{\pm{0.24}}$ & 8.7\%\\
\texttt{Anime}  & 7.10$_{\pm{0.18}}$ & 9.77$_{\pm{0.17}}$ & 7.88$_{\pm{0.21}}$ & 8.03$_{\pm{0.12}}$ & 1.9\%& 8.48$_{\pm{0.07}}$ & 9.68$_{\pm{0.03}}$ & 14.2\%& 9.98$_{\pm{0.16}}$ & 10.69$_{\pm{0.24}}$ & 7.0\%\\
\texttt{Gowalla}  & 5.54$_{\pm{0.07}}$ & 6.59$_{\pm{0.12}}$ & 2.54$_{\pm{0.17}}$ & 2.68$_{\pm{0.09}}$ & 5.4\%& 6.54$_{\pm{0.07}}$ & 6.73$_{\pm{0.04}}$ & 2.9\%& 5.38$_{\pm{0.10}}$ & 5.70$_{\pm{0.11}}$ & 6.0\%\\
\texttt{Yelp-2018}   & 2.80$_{\pm{0.07}}$ & 3.51$_{\pm{0.10}}$ & 1.80$_{\pm{0.08}}$ & 1.89$_{\pm{0.04}}$ & 5.2\%& 3.52$_{\pm{0.09}}$ & 3.57$_{\pm{0.05}}$ & 1.6\%& 2.82$_{\pm{0.10}}$ & 3.16$_{\pm{0.14}}$ & 12.1\%\\
\texttt{MovieLens-1M}  & 15.14$_{\pm{0.30}}$ & 17.81$_{\pm{0.32}}$ & 12.55$_{\pm{0.19}}$ & 15.55$_{\pm{0.22}}$ & 23.9\%& 14.44$_{\pm{0.11}}$ & 20.11$_{\pm{0.17}}$ & 39.2\%& 16.39$_{\pm{0.17}}$ & 19.54$_{\pm{0.23}}$ & 19.2\%\\
\cmidrule(r){1-12} 
\multicolumn{12}{c}{\textsc{NDCG@10 -- Overall}} \\
\cmidrule(r){1-12} 
\texttt{Amazon-Book}  & 4.35$_{\pm{0.13}}$ & 5.03$_{\pm{0.10}}$ & 3.81$_{\pm{0.11}}$ & 4.09$_{\pm{0.08}}$ & 7.4\%& 4.98$_{\pm{0.03}}$ & 5.08$_{\pm{0.03}}$ & 1.9\%& 3.60$_{\pm{0.25}}$ & 3.83$_{\pm{0.23}}$ & 6.6\%\\
\texttt{Anime}  & 7.99$_{\pm{0.27}}$ & 9.92$_{\pm{0.24}}$ & 10.63$_{\pm{0.10}}$ & 10.97$_{\pm{0.11}}$ & 3.1\%& 8.52$_{\pm{0.05}}$ & 9.73$_{\pm{0.03}}$ & 14.2\%& 10.66$_{\pm{0.10}}$ & 10.99$_{\pm{0.09}}$ & 3.1\%\\
\texttt{Gowalla}  & 5.65$_{\pm{0.11}}$ & 6.54$_{\pm{0.11}}$ & 3.42$_{\pm{0.05}}$ & 3.49$_{\pm{0.05}}$ & 2.1\%& 6.41$_{\pm{0.07}}$ & 6.53$_{\pm{0.04}}$ & 1.9\%& 5.59$_{\pm{0.14}}$ & 5.93$_{\pm{0.14}}$ & 6.1\%\\
\texttt{Yelp-2018}  & 3.19$_{\pm{0.05}}$ & 3.64$_{\pm{0.05}}$ & 2.18$_{\pm{0.11}}$ & 2.26$_{\pm{0.05}}$ & 3.6\%& 3.59$_{\pm{0.07}}$ & 3.67$_{\pm{0.03}}$ & 2.3\%& 2.88$_{\pm{0.11}}$ & 3.21$_{\pm{0.10}}$ & 11.5\%\\
\texttt{MovieLens-1M}  & 15.94$_{\pm{0.20}}$ & 18.24$_{\pm{0.27}}$ & 14.12$_{\pm{0.17}}$ & 15.80$_{\pm{0.21}}$ & 11.9\%& 15.47$_{\pm{0.11}}$ & 20.51$_{\pm{0.18}}$ & 32.6\%& 18.23$_{\pm{0.13}}$ & 20.43$_{\pm{0.25}}$ & 12.1\%\\
\midrule
\midrule
\multicolumn{12}{c}{\textsc{Recall@10 -- Low-degree Users (Lower Percentile)}} \\
\cmidrule(r){1-12} 
\texttt{Amazon-Book} & 3.75$_{\pm{0.12}}$ & 4.64$_{\pm{0.19}}$ & 3.65$_{\pm{0.14}}$ & 3.93$_{\pm{0.13}}$ & 7.6\%& 4.57$_{\pm{0.10}}$ & 4.74$_{\pm{0.10}}$ & 3.8\%& 2.77$_{\pm{0.15}}$ & 2.93$_{\pm{0.12}}$ & 5.8\%\\
\texttt{Anime}  & 10.10$_{\pm{0.28}}$ & 12.84$_{\pm{0.24}}$ & 14.64$_{\pm{0.64}}$ & 15.15$_{\pm{0.57}}$ & 3.5\%& 11.41$_{\pm{0.08}}$ & 12.61$_{\pm{0.06}}$ & 10.5\%& 13.34$_{\pm{0.25}}$ & 14.52$_{\pm{0.30}}$ & 8.8\%\\
\texttt{Gowalla}  & 7.22$_{\pm{0.28}}$ & 7.88$_{\pm{0.17}}$ & 3.59$_{\pm{0.06}}$ & 3.74$_{\pm{0.07}}$ & 4.2\%& 7.82$_{\pm{0.21}}$ & 7.96$_{\pm{0.13}}$ & 1.8\%& 6.40$_{\pm{0.15}}$ & 6.66$_{\pm{0.17}}$ & 4.2\%\\
\texttt{Yelp-2018}   & 3.45$_{\pm{0.12}}$ & 3.64$_{\pm{0.13}}$ & 2.44$_{\pm{0.06}}$ & 2.57$_{\pm{0.10}}$ & 5.5\%& 3.60$_{\pm{0.12}}$ & 3.73$_{\pm{0.16}}$ & 3.6\%& 2.62$_{\pm{0.16}}$ & 2.92$_{\pm{0.19}}$ & 11.5\%\\
\texttt{MovieLens-1M}  & 6.60$_{\pm{0.15}}$ & 7.54$_{\pm{0.18}}$ & 5.73$_{\pm{0.11}}$ & 7.07$_{\pm{0.14}}$ & 23.4\%& 6.90$_{\pm{0.18}}$ & 8.22$_{\pm{0.22}}$ & 19.2\%& 7.66$_{\pm{0.22}}$ & 8.54$_{\pm{0.23}}$ & 11.5\%\\
\cmidrule(r){1-12} 
\multicolumn{12}{c}{\textsc{Recall@10 -- Overall}} \\
\cmidrule(r){1-12} 
\texttt{Amazon-Book}  & 3.62$_{\pm{0.22}}$ & 4.46$_{\pm{0.17}}$ & 3.81$_{\pm{0.17}}$ & 3.99$_{\pm{0.08}}$ & 4.8\%& 4.44$_{\pm{0.06}}$ & 4.56$_{\pm{0.07}}$ & 2.6\%& 2.81$_{\pm{0.24}}$ & 3.00$_{\pm{0.26}}$ & 6.7\%\\
\texttt{Anime}  & 11.12$_{\pm{0.18}}$ & 12.86$_{\pm{0.22}}$ & 13.45$_{\pm{0.26}}$ & 13.71$_{\pm{0.26}}$ & 1.9\%& 11.43$_{\pm{0.08}}$ & 12.71$_{\pm{0.04}}$ & 11.2\%& 14.19$_{\pm{0.46}}$ & 14.64$_{\pm{0.39}}$ & 3.2\%\\
\texttt{Gowalla}  & 7.43$_{\pm{0.06}}$ & 7.70$_{\pm{0.12}}$ & 4.01$_{\pm{0.07}}$ & 4.04$_{\pm{0.08}}$ & 0.8\%& 7.57$_{\pm{0.18}}$ & 7.72$_{\pm{0.09}}$ & 2.0\%& 6.56$_{\pm{0.17}}$ & 6.77$_{\pm{0.21}}$ & 3.4\%\\
\texttt{Yelp-2018}  & 3.39$_{\pm{0.06}}$ & 3.73$_{\pm{0.09}}$ & 2.32$_{\pm{0.09}}$ & 2.41$_{\pm{0.03}}$ & 3.7\%& 3.67$_{\pm{0.12}}$ & 3.80$_{\pm{0.11}}$ & 3.6\%& 2.83$_{\pm{0.17}}$ & 3.36$_{\pm{0.24}}$ & 18.7\%\\
\texttt{MovieLens-1M}  & 7.00$_{\pm{0.13}}$ & 7.66$_{\pm{0.20}}$ & 6.20$_{\pm{0.23}}$ & 7.02$_{\pm{0.25}}$ & 13.2\%& 7.68$_{\pm{0.12}}$ & 8.33$_{\pm{0.16}}$ & 8.4\%& 7.93$_{\pm{0.20}}$ & 8.75$_{\pm{0.22}}$ & 10.3\%\\
\bottomrule
\end{tabular}
}
\label{tab:rebuttal_table1}
\end{center}
\end{table*}

\begin{figure}
    \includegraphics[width=1\textwidth]{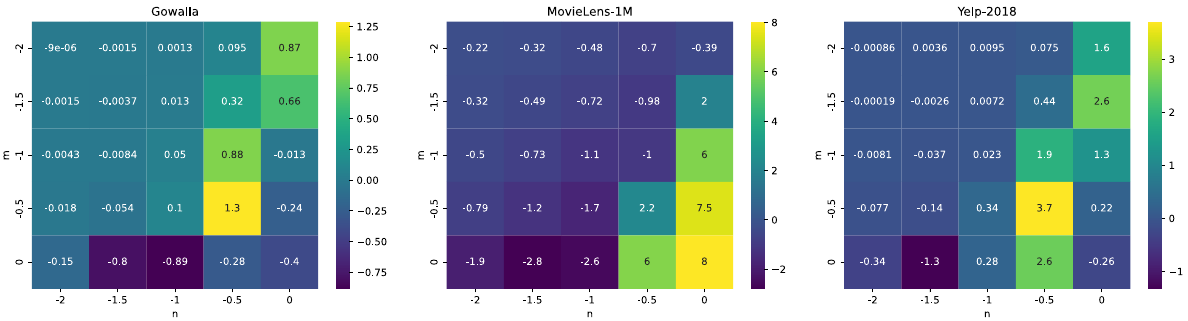}
    \caption{The sensitivity of TAG-CF to $m$ and $n$ in \cref{eq:tagcf}. Numbers reported in these plots are performance improvement (\%) brought by TAG-CF to MF trained by DirectAU~\citep{wang2022towards} on Recall@20. }
    \label{fig:mandn}
    \vspace{-0.2in}
\end{figure}


\newpage
\section*{NeurIPS Paper Checklist}

\begin{enumerate}

\item {\bf Claims}
    \item[] Question: Do the main claims made in the abstract and introduction accurately reflect the paper's contributions and scope?
    \item[] Answer: \answerYes{} 
    \item[] Justification: This paper discusses how and why does message passing help collaborative filtering. We approach this question by analyzing from two perspectives (e.g., Section 3) and propose TAG-CF given our findings. 

\item {\bf Limitations}
    \item[] Question: Does the paper discuss the limitations of the work performed by the authors?
    \item[] Answer: \answerYes{} 
    \item[] Justification: We discussed limitations at the very end of the paper.

\item {\bf Theory Assumptions and Proofs}
    \item[] Question: For each theoretical result, does the paper provide the full set of assumptions and a complete (and correct) proof?
    \item[] Answer: \answerYes{} 
    \item[] Justification: We propose a theorem in our main paper and provide proofs in the appendix (i.e., Appendix 1). 

    \item {\bf Experimental Result Reproducibility}
    \item[] Question: Does the paper fully disclose all the information needed to reproduce the main experimental results of the paper to the extent that it affects the main claims and/or conclusions of the paper (regardless of whether the code and data are provided or not)?
    \item[] Answer: \answerYes{} 
    \item[] Justification: We provide hyper-parameter setups to reproduce our experiments (i.e., Appendix C.2). Besides, we also provide source code to reproduce our experiments. 

\item {\bf Open access to data and code}
    \item[] Question: Does the paper provide open access to the data and code, with sufficient instructions to faithfully reproduce the main experimental results, as described in supplemental material?
    \item[] Answer: \answerYes{} 
    \item[] Justification: We append our code for others to reproduce our results. 

\item {\bf Experimental Setting/Details}
    \item[] Question: Does the paper specify all the training and test details (e.g., data splits, hyperparameters, how they were chosen, type of optimizer, etc.) necessary to understand the results?
    \item[] Answer: \answerYes{} 
    \item[] Justification: Training and test details are specified in C1 and hyper-parameter tuning strategies are specified in C2. 

\item {\bf Experiment Statistical Significance}
    \item[] Question: Does the paper report error bars suitably and correctly defined or other appropriate information about the statistical significance of the experiments?
    \item[] Answer: \answerYes{} 
    \item[] Justification: We run our experiments 10 times and report both mean and standard deviation of numbers we report. 

\item {\bf Experiments Compute Resources}
    \item[] Question: For each experiment, does the paper provide sufficient information on the computer resources (type of compute workers, memory, time of execution) needed to reproduce the experiments?
    \item[] Answer: \answerYes{} 
    \item[] Justification: Resources used are specified in Appendix C.3.
    
\item {\bf Code Of Ethics}
    \item[] Question: Does the research conducted in the paper conform, in every respect, with the NeurIPS Code of Ethics \url{https://neurips.cc/public/EthicsGuidelines}?
    \item[] Answer: \answerYes{} 
    \item[] Justification: We conform the code Of Ethics and have some discussion at the very end of our paper. 

\item {\bf Broader Impacts}
    \item[] Question: Does the paper discuss both potential positive societal impacts and negative societal impacts of the work performed?
    \item[] Answer: \answerYes{} 
    \item[] Justification: We observe no ethical concern entailed by our proposal, but we note that both ethical or unethical applications based on collaborative filtering may benefit from the effectiveness of our work. 
Care should be taken to ensure socially positive and beneficial results of machine learning algorithms.
    
\item {\bf Safeguards}
    \item[] Question: Does the paper describe safeguards that have been put in place for responsible release of data or models that have a high risk for misuse (e.g., pretrained language models, image generators, or scraped datasets)?
    \item[] Answer: \answerNA{} 
    \item[] Justification: \answerNA{}

\item {\bf Licenses for existing assets}
    \item[] Question: Are the creators or original owners of assets (e.g., code, data, models), used in the paper, properly credited and are the license and terms of use explicitly mentioned and properly respected?
    \item[] Answer: \answerYes{} 
    \item[] Justification: With in the folder of our code, we include the license of all code, data, and tools we use. 

\item {\bf New Assets}
    \item[] Question: Are new assets introduced in the paper well documented and is the documentation provided alongside the assets?
    \item[] Answer: \answerYes{} 
    \item[] Justification: We have docstrings for all functions in our code and we also provide a readme file to help others use our code. 

\item {\bf Crowdsourcing and Research with Human Subjects}
    \item[] Question: For crowdsourcing experiments and research with human subjects, does the paper include the full text of instructions given to participants and screenshots, if applicable, as well as details about compensation (if any)? 
    \item[] Answer: \answerNA{}{} 
    \item[] Justification: \answerNA{}

\item {\bf Institutional Review Board (IRB) Approvals or Equivalent for Research with Human Subjects}
    \item[] Question: Does the paper describe potential risks incurred by study participants, whether such risks were disclosed to the subjects, and whether Institutional Review Board (IRB) approvals (or an equivalent approval/review based on the requirements of your country or institution) were obtained?
    \item[] Answer: \answerNA{} 
    \item[] Justification: \answerNA{}

\end{enumerate}

\end{document}